%% file: tsp11_nonbayesian.tex
\documentclass[11pt,draftcls,onecolumn]{IEEEtran} 
\usepackage{graphics}
\usepackage{epsfig,subfigure}
\usepackage{amsmath}
\usepackage{amssymb}
\usepackage{epsfig}
\usepackage{psfrag}
\usepackage{cite}
\usepackage[usenames]{color}

\newcommand{\textr}[1]{#1}
\newcommand{\textb}[1]{#1}

\newtheorem{theorem}{Theorem}
\newtheorem{definition}{Definition}

\renewcommand{\eqref}[1]{(\ref{eq:#1})}
\newcommand{\secref}[1]{Sec.~\ref{sec:#1}}

\newcommand{\figref}[1]{Fig.~\ref{fig:#1}}

\newcommand{\defn}{\triangleq}
\renewcommand{\vec}[1]{\mathbf{#1}}
\newcommand{\probe}{_{\text{\sf p}}}
\newcommand{\data}{_{\text{\sf d}}}
\newcommand{\tot}{_{\text{\sf sum}}}
\newcommand{\genie}{^{\text{\sf genie}}}

\newcommand{\E}[1]{{\mbox{E}}\left[{#1}\right]}
\newcommand{\Ep}[1]{{\mbox{E}_p}\left[{#1}\right]}

\newcommand{\Eq}[1]{{\mbox{E}_q}\left[{#1}\right]}
\newcommand{\Eqsmall}[1]{{\mbox{E}_q}\big[{#1}\big]}
\newcommand{\var}[1]{{\mbox{var}}\left({#1}\right)}
\newcommand{\varsmall}[1]{{\mbox{var}}\big({#1}\big)}

\newcommand{\varqsmall}[1]{{\mbox{var}_q}\big({#1}\big)}

\newcommand{\Prob}[1]{{\mbox{Pr}}\bigl({#1}\bigr)}

\newcommand{\comment}[1]{}

\DeclareMathOperator*{\argmax}{argmax}

\date{\today}

\title{\textr{Robust} Rate-Adaptive Wireless Communication Using ACK/NAK-Feedback}
\author{C. Emre Koksal and Philip Schniter%
        \thanks{The authors are with the Dept.\ of Electrical
                and Computer Engineering
                at The Ohio State University, Columbus, OH 43210.
                Please direct all correspondence to 
		Prof.~C.~Emre Koksal, 
		Dept.~ECE, 2015 Neil Ave., Columbus OH 43210,
		phone 614.688.4369,
		fax 614.292.7596,
		and e-mail koksal@ece.osu.edu.
		Philip Schniter
		can be reached at the same address and fax, 
                phone 614.247.6488, and e-mail schniter@ece.osu.edu.
		}%
}

\begin{document}

\maketitle

\input{abstract}
\input{intro}
\input{model}

\input{necessary}
\input{capacity}
\input{estimator}
\input{conclusion}
\bibliographystyle{ieeetr}
\bibliography{macros_abbrev,books,comm,misc,net}

\appendices
\input{training}
\input{estimator_convergence}

\end{document}

%% file: abstract.tex
\begin{abstract}
To combat the detrimental effects of the variability in wireless channels, we
consider cross-layer rate adaptation based on limited feedback. In particular,
based on limited feedback in the form of link-layer acknowledgements (ACK) and
negative acknowledgements (NAK), we maximize the physical-layer transmission
rate subject to an upper bound on the expected packet error rate.
\textr{We take a robust approach in that we do not assume} any
particular prior distribution on the channel state. We first analyze the
fundamental limitations of such systems and derive an upper bound on the
achievable rate for signaling schemes based on uncoded QAM and random
Gaussian ensembles. We show that, for channel estimation based on binary
ACK/NAK feedback, it may be preferable to use a separate training sequence at
high error rates, rather than to exploit low-error-rate data packets
themselves. We also develop an adaptive recursive estimator, which is provably
asymptotically optimal and asymptotically efficient.

\vspace{2mm} {{\bf Index Terms}}---\,%
  adaptive modulation,
  rate adaptation,
  automatic repeat request, 
  cross-layer strategies.
\end{abstract}

%% file: intro.tex
\section{Introduction}
\label{sec:intro}

Channel variation is a principal feature of wireless communication.
On one hand, channel variation poses a hindrance to reliable communication, in that channel fading can make the received signal-to-noise ratio (SNR) arbitrarily low at any given time instant, making reliable communication virtually impossible.
On the other hand, channel variation poses an opportunity, in that a channel-state-aware transmitter can communicate reliably at high rates during channel quality peaks.
The key to taming and exploiting channel variation therefore lies in the judicious use of transmitter channel state information (CSI).
While accurate \emph{receiver} CSI is relatively easy to maintain, accurate \emph{transmitter} CSI is often difficult to maintain due to limited feedback resources.

We partition limited feedback schemes (see~\cite{Heath:JSAC:08} for an overview) into two classes: those based on \emph{channel-state feedback} and those based on \emph{error-rate feedback}.
In limited channel-state feedback schemes (e.g., \cite{Goldsmith:Book:05, Goldsmith:TCOM:97,Goeckel:TCOM:99,Balachandran:JSAC:99}), the channel-state estimate computed by the receiver is quantized\footnote{
  In some cases, the receiver uses its channel estimate to calculate discrete transmitter rate and/or power parameters, and then feeds back those parameters directly.  Since these transmitter parameters can be put in one-to-one correspondence with some quantized channel-state estimate, we consider such schemes to be equivalent to channel-state feedback schemes.}
and then fed back to the transmitter.
In limited error-rate feedback schemes (e.g., \cite{Holland:MOBICOM:01,Sadegi:MOBICOM:02,Bicket:Thesis:05,Wong:MOBICOM:06,Minn:TVT:01,Rice:TCOM:94,Yao:TCOM:95,Chakraborty:COML:99,Choi:TVT:01,Karmokar:TWC:06,Djonin:TVT:08,Aggarwal:TWC:09}), a quantized error-rate estimate is fed back to the transmitter, from which it can infer CSI \emph{relative to} the previously employed transmission rate.
For example, with Automatic Repeat reQuest (ARQ) \cite{Bertsekas:Book:92}, a negative acknowledgement (NAK) of packet reception suggests that the channel quality was below that needed for reliable communication at the previously employed transmission rate, whereas a positive acknowledgement (ACK) of packet reception suggests the opposite.

Although ACK/NAK feedback can be employed for the estimation of transmitter CSI, its primary role is that of maintaining a desired packet error rate at the link layer through controlled packet re-transmission (see, e.g.,~\cite{Bertsekas:Book:92}).
In fact, since the packet acknowledgement is a standard provision of most practical link layers, we reason that---for the purpose of channel-state estimation---it comes at \emph{essentially no cost} to the physical layer, unlike traditional channel-state feedback schemes, which require the dedication of reverse-channel bandwidth beyond that required for packet acknowledgements.
In this sense, ACK/NAK-based transmitter-CSI schemes require even less total feedback bandwidth than ``one-bit'' channel-state feedback schemes (e.g.,~\cite{Hassel:TWC:07,Rong:TCOM:06}), given that systems employing ``one-bit'' channel-state feedback include ACK/NAK as well, for the purpose of ARQ.

With the above motivation, we focus on the \emph{exclusive} use of limited error-rate feedback for the maintenance of transmitter CSI, from which transmission rate and/or power resources are subsequently adapted.
While examples of this strategy can be found in a number of previous works (e.g., \cite{Holland:MOBICOM:01,Sadegi:MOBICOM:02,Bicket:Thesis:05,Wong:MOBICOM:06,Minn:TVT:01,Rice:TCOM:94,Yao:TCOM:95,Chakraborty:COML:99,Choi:TVT:01,Karmokar:TWC:06,Djonin:TVT:08,Aggarwal:TWC:09}), there are limitations in how it has been applied. 
For example, in \cite{Holland:MOBICOM:01,Sadegi:MOBICOM:02,Bicket:Thesis:05,Wong:MOBICOM:06,Minn:TVT:01}, the adaptation algorithms are designed heuristically, based on practical experiences gained for a specific application in a specific operating environment. 
In~\cite{Rice:TCOM:94,Yao:TCOM:95,Chakraborty:COML:99,Choi:TVT:01,Karmokar:TWC:06,Djonin:TVT:08,Aggarwal:TWC:09}, on the other hand, transmission rates and/or powers are chosen carefully to maximize a certain performance metric. 
To achieve this objective, a Bayesian approach is taken, i.e., a {\em model} is assumed for the channel variations and an associated optimization problem is solved based on this model. 
Typically, the channel is assumed to vary according to a finite-state Markov model~\cite{Rice:TCOM:94,Yao:TCOM:95,Choi:TVT:01,Karmokar:TWC:06,Djonin:TVT:08} or a Gauss-Markov process~\cite{Aggarwal:TWC:09}. 
The shortcoming of a model-based approach is that, it may not be possible to assign accurate priors over a wide range of channel operating conditions.
Consider, for example, that channel variations span a wide range of time scales, from bits to thousands of packets.
For instance, relative movement of the transmitter-receiver pair may cause variations at relatively long time scales, since a very large number of packets can be transmitted during the time it takes for the stations to move far enough to cause significant change in the channel.
On the other hand, co-channel interference can change significantly from one packet transmission to another.
Finally, the multipath nature of the propagation medium can cause fast and/or slow fading in the channel, depending on the relative movement of the scatterers.

In this paper, 
\textb{we take a robust Bayesian \cite{Berger:Book:85} approach to rate-adaptation from limited error-rate feedback, where ``robust Bayesian'' refers to the fact that we treat the channel state as a random quantity without assuming any particular prior distribution on it.}
In particular, we first derive conditions on the ``quality'' of CSI needed for a model-independent ACK/NAK-based rate adaptation system to maximize data rate while keeping the packet error probability below a specified threshold.
Based on these conditions, we derive fundamental bounds on the rate achievable under a given error probability constraint. Finally, we design an ACK/NAK-feedback-based \textb{non-Bayesian channel-state estimator} with provable asymptotic optimality.
Our findings are illustrated through both uncoded QAM and random Gaussian signaling.

We emphasize that the packet-level retransmissions orchestrated by link-layer ARQ would be performed \emph{on top of} the ACK/NAK-based rate-control that we study.
In fact, since our physical-layer optimization criterion (i.e., maximization of transmission rate subject to a given target packet error probability) is by nature decoupled from the functioning of higher layers, we do not explicitly consider ARQ in our analysis.
In other words, from the perspective of our physical layer, the link-layer ARQ mechanism merely specifies the contents of the packets that are to be transmitted.

The remainder of the paper is organized as follows.
In Section~\ref{sec:model}, we detail the system model and provide a mathematical statement of the problem.
In Section~\ref{sec:necessary_condition}, we derive conditions for successful rate adaptation with imperfect CSI, and in Section~\ref{sec:capacity}, we evaluate bounds on the achievable rates with ACK/NAK feedback.
In Section~\ref{sec:estimator}, we develop an recursive channel estimator based on such feedback, and in Section~\ref{sec:conc} we conclude.

%% file: model.tex
\section{System Model}
\label{sec:model}

\subsection{System Components}
\label{sec:components}

\begin{figure}
\begin{center}
\psfrag{controller}[][Bl][0.7]{\sf \begin{tabular}{c}rate\\[-0mm]controller\end{tabular}}
\psfrag{encoder}[][Bl][0.7]{\sf encoder}
\psfrag{forward}[][Bl][0.7]{\sf \begin{tabular}{c}forward\\[-0mm]channel\end{tabular}}
\psfrag{decoder}[][Bl][0.7]{\sf decoder}
\psfrag{feedback}[][Bl][0.7]{\sf feedback}
\psfrag{reverse}[][Bl][0.7]{\sf \begin{tabular}{c}reverse\\[-0mm]channel\end{tabular}}
\psfrag{data}[][Bl][0.7]{\sf data}
\psfrag{gam}[B][Bl][0.8]{$H_t,\gamma_t$}
\psfrag{Ftd}[B][Bl][0.8]{$F_{t-1}$~}
\psfrag{Rt}[B][Bl][0.8]{$R_t$}
\psfrag{Xt}[B][Bl][0.8]{${\bf X}_t$}
\psfrag{Yt}[B][Bl][0.8]{${\bf Y}_t$}
\psfrag{Xh}[B][Bl][0.8]{$\hat{\bf X}_t$}
\psfrag{eps}[][Bl][0.7]{$\varepsilon(\gamma_t,R_t)$}
\psfrag{Ft}[B][Bl][0.8]{$F_t$}
\includegraphics[height=1.25in]{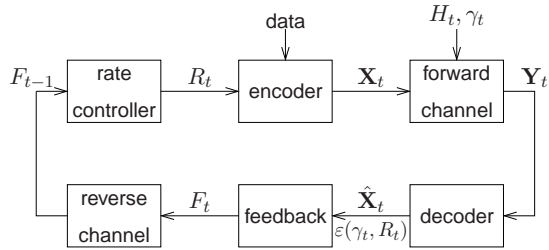}
\end{center}
\caption{The rate adaptation system.}
\label{fig:system}
\end{figure}

Figure~\ref{fig:system} depicts our model of the physical-layer
adaptive communication system. At each discrete packet index $t$, the
transmitter transmits a packet ${\bf X}_t=[X_{t,1},\dots,X_{t,n}]$
containing a fixed number, $n$, of symbols $\{X_{t,k}\}_{k=1}^n$, which are
encoded at a rate of $R_t$ bits/symbol, chosen by the rate
controller from the set of possible rates ${\mathcal R}$. We assume
that the transmit power is constant and normalize all power levels
such that the energy per symbol is $\E{|X_{t,k}|^2}=1$.
For this packet, the corresponding channel outputs are
\begin{equation}
\label{eq:channel_output}
Y_{t,k}=H_tX_{t,k}+W_{t,k},~~k=1,\dots,n,
\end{equation}
for complex-valued channel gain $H_t$ and additive white circularly symmetric
complex Gaussian noise $W_{t,k}$ with two-sided power spectral density $N_o$.
Some common models for $H_t$ include
Rayleigh-, Rician- and Nakagami-fading (see e.g.,~\cite{Simon:Book:00}).
However, we will not assume any specific statistical model for $H_t$
and we will make only weak assumptions on the distribution of
$H_t$ in the sequel.


The quantity $\gamma_t=|H_t|^2 / N_o$ can be interpreted as the $t^{th}$
packet's {\emph{channel SNR}}. Since each symbol has unit energy, $\gamma_t$ is also the {\emph{received SNR}} for packet $t$. Thus, we will simply refer to $\gamma_t$ as the SNR. Due to lack of power adaptation, $\gamma_t$ is an exogenous quantity over which the system has no control.
We assume that, for all $t$, $\gamma_t$ takes on values from some
prior distribution $p(\cdot ) \in {\mathcal P}$, where ${\mathcal
P}$ is a set of distributions with finite mean and variance. \textr{
However, we make no further assumptions on set ${\cal P}$. We do not even
assume knowledge of this set by the transmitter or the receiver.}

We assume that the receiver has access to perfect CSI and uses a maximum likelihood decoder to decode the received packet. Let $\hat{\bf X}_t$ denote the decoded estimate of packet ${\bf X}_t$ based on received packet ${\bf Y}_t=[Y_{t,1},\dots,Y_{t,n}]$, and the corresponding probability of decoding error be
$\varepsilon(\gamma_t,R_t)=\Prob{{\bf{\hat{X}}}_t\neq {\bf{X}}_t\
|\ \gamma_t,R_t}$.
Note that $\varepsilon(\cdot,\cdot)$ depends on
the packet size $n$ and the
coding/modulation schemes, which are assumed to be known at the decoder.
For now, we assume only that the coding/modulation schemes
are such that $\varepsilon(\gamma_t,R_t)$ is a convex, continuous,
and increasing function of $R_t$ and a convex, continuous, and decreasing
function of $\gamma_t$.
Later, we detail the behavior of our proposed schemes for the specific cases of uncoded QAM and random Gaussian signaling.

Based on the received packet ${\bf Y}_t$ and the decoded packet
$\hat{\bf X}_t$, the decoder generates a feedback packet $F_t$ which
is communicated to the transmitter through a reverse channel.
Assuming that the receiver is capable of perfect error detection, we
take $F_t$ to be a binary ACK/NAK (i.e., $F_t=0$ for ACK and $F_t=1$ for NAK),
so that
\begin{equation}
\label{eq:ack_nack_prob} \Prob{F_t=f\ |\ \gamma_t,R_t}=\begin{cases}
\varepsilon(\gamma_t,R_t), & f=1 \\
1-\varepsilon(\gamma_t,R_t), & f=0 \end{cases} .
\end{equation}
We assume that the reverse channel is error-free but introduces a
delay of a single\footnote{It is straightforward to generalize all of our results to a general delay of $d\geqslant 1$ packet intervals.
While the generalization does not alter the fundamental nature of our results, it requires a more complex notation, which we avoid for clarity.} packet interval.
Thus, the ``information'' available to the transmitter when choosing rate $R_t$ is ${\bf I}_t=[F_1,F_2,\ldots , F_{t-1},R_1,R_2,\ldots,R_{t-1}]$.
We find it convenient to explicitly include the previous rates
$\{R_\tau\}_{\tau<t}$ in the information vector ${\bf I}_t$ because
the ACK/NAK feedback $F_\tau$ characterizes channel quality
\emph{relative to} the transmission rate $R_\tau$.
Note that the controller chooses the transmission rate at time $t$ solely based
on the information vector ${\bf I}_t$, which is available
at the receiver as well.
We assume that the receiver is also aware
of the controller's rate allocation strategy, so that it
can compute the current and previous values of $R_t$.

Finally, we assume in the sequel that the SNR is constant over each
block of $T\gg 1$ packets, and that it changes independently from
block to block, i.e., that the channel is ``block fading.'' In the
sequel, we focus (without loss of generality) on the first block,
for which $t\in\{1,\ldots,T\}$, and omit the $t$-dependence on the
SNR, writing $\gamma_t$ as ``$\gamma$.'' In addition, we use
$p(\gamma|{\bf I}_t)$ to denote the posterior SNR distribution,
which can be associated with the prior distribution $p(\gamma)$
through the conditional mass function $P(F_t\ |\ \gamma,R_t)$ given
in (\ref{eq:ack_nack_prob}). Furthermore, we denote the set of
possible posterior probability distributions using
${\mathcal{P}}({\bf I}_t)$.


\subsection{Ideal Rate Selection}
\label{sec:ideal}

We define the \textr{{\em ideal $p$-hypothesized controller}} as the
one that, at time $t$, based on the \textr{hypothesized posterior
$p(\gamma|{\bf I}_t)$}, jointly optimizes the
transmission rates $(R_t,\dots,R_T)$ to maximize the sum-rate
$\sum_{\tau=1}^T R_\tau$ subject to a constraint on expected error probability.
In doing so, we allow any packet to be declared a \emph{probe packet},
which is exempt from the expected-error-probability constraint but
contributes nothing to sum rate.
Probe packets are used exclusively to learn about the SNR $\gamma$, in
the hope of more efficient allocation of future \emph{data packets}.
In particular, the ideal controller chooses rates according to the following
constrained optimization problem:
\begin{align}
\label{eq:objective}
& \max_{(D_t,\ldots,D_T) \in \{0,1\}^{T-t+1},~
(R_t,\ldots,R_T) \in {\mathcal R}^{T-t+1}} \quad
\sum_{\tau=t}^T  D_\tau R_\tau \\
\label{eq:constraint}
& \text{subject to~} 
D_\tau \Ep{\varepsilon(\gamma,R_\tau)\ |\ {\bf I}_t}
\leqslant e^{-\alpha}
\text{~for all}\ \tau=t,\ldots ,T .
\end{align}
Here, $D_\tau\in\{0,1\}$ indicates whether the $\tau^{th}$ packet is a
data packet ($D_\tau=1$) or a probe packet ($D_\tau=0$), and
$\alpha>0$ is an application-dependent quality-of-service (QoS) parameter.
Note that the expectation $\Ep{\cdot}$ in (\ref{eq:constraint})
is taken over the conditional distribution $p(\gamma|{\bf I}_t)$.

With ACK/NAK feedback, recall that
${\bf I}_t=[F_1,F_2,\ldots,F_{t-1},R_1,R_2,\ldots,R_{t-1}]$.
Thus, the choice of $R_t$ affects not only the
contribution to the sum-rate but also the ``quality'' of the conditional SNR distribution $p(\gamma~|~{\bf I}_{\tau})$ at times $\tau\geq t+1$.
As these future SNR estimates get worse, the controller is forced to
choose more conservative (i.e., lower) rates in order to satisfy the
expected error-rate constraint.
(We justify this statement in the sequel.)
Thus, the selection of $R_t$ has both short-term and long-term consequences, which may be in conflict.
Consequently, the solution to the ideal rate adaptation
problem (\ref{eq:objective},\ref{eq:constraint}) under ACK/NAK feedback is a
{\em constrained} partially observable Markov decision process (POMDP)~\cite{Monahan:MS:82}.
For practical horizons $T$, it is computationally impractical to
implement this POMDP, as now described.  Firstly, notice that the
state of the channel is continuous. Even if the channel state was
discretized (at the expense of some loss in performance), the required
memory to implement the optimal scheme would grow exponentially with the
horizon $T$. Indeed, this POMDP lies in the space of PSPACE-complete
problems, i.e., it requires both complexity and memory that grow
exponentially with the horizon $T$~\cite{Papadimitriou:MOR:87}.

Next, consider the (genie-aided) case of perfect CSI, i.e.,
${\bf I}_t=\gamma$ for all $t$.
When the channel is known, there is no need for probe packets,
and thus the optimal solution chooses $D_\tau=1~\forall \tau$.
Furthermore, since the rate choice does not affect the quality
of the SNR estimate, the ideal rate assignment problem decouples,
so that the best choice for $R_t$ becomes
\begin{equation}
\label{eq:genie}
R_t^{\text{perf-CSI}}(\gamma)
~\defn~ \arg\max_{R_t\in\mathcal{R}} R_t \text{~~s.t.~~}
\varepsilon(\gamma,R_t) \leqslant e^{-\alpha} .
\end{equation}
Indeed, with perfect CSI,
constraint~(\ref{eq:constraint}) is active for all $t=1,\ldots
,T$, since $\varepsilon(\gamma,R_t)$ is a convex increasing
function of $R_t$ and the objective function is linear in $R_t$.
Notice that, in this case, ideal rate selection is greedy and
$R_t^{\text{perf-CSI}}(\gamma)$ is invariant\footnote{
  This invariance holds as long as $\varepsilon(\cdot,\cdot)$ is
  $t$-invariant, i.e., the coding/modulation scheme does not change with time.}
to time $t$.

\subsection{Practical Rate Selection}
\label{sec:practical}

\begin{figure}
\psfrag{I}[r][Bl][0.8]{${\bf I}_{T\probe +1}$}
\psfrag{gam}[b][Bl][0.8]{$\Hat{\gamma}({\bf I}_{T\probe +1})$}
\psfrag{R}[l][Bl][0.8]{$\{ R_{T\probe+1},\ldots ,R_T \}$}
\psfrag{x}[l][Bl][0.8]{}
\psfrag{ce}[][Bl][0.8]{\sf\begin{tabular}{c}channel\\[-0mm]estimator\end{tabular}}
\psfrag{ra}[][Bl][0.8]{\sf\begin{tabular}{c}rate\\[-0mm]allocator\end{tabular}}
\centerline{\includegraphics[height=0.32in]{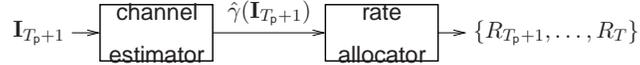}}
\caption{The controller decomposed into two components: a channel estimator and a rate allocator.}
\label{fig:estimator_controller}
\end{figure}
\textr{In practice, we have neither the exact posterior
$p(\gamma|{\bf I}_t)$, nor the perfect CSI. Thus,} we consider a practical
(non-ideal) approach, motivated by techniques from the field of adaptive
control~\cite{Kumar:Book:86}, which deviates from the ideal approach in two
principal ways:
\begin{enumerate}
\item
the probe packet locations are set at the first $T\probe$ packets in each
$T$-block, and
\item
the controller is split into two components: a \emph{channel estimator},
which produces an SNR estimate $\Hat{\gamma}({\bf I}_{T\probe+1})$ based on
the probe-packet feedback ${\bf I}_{T\probe+1}$, and a \emph{rate allocator},
which assigns the data packet rate based on $\Hat{\gamma}({\bf I}_{T\probe+1})$.
(See Fig.~\ref{fig:estimator_controller}.)
\end{enumerate}
As before, the rate allocator chooses the data-packet rates
$(R_{T\probe+1},\ldots,R_T)$ in order to maximize sum-rate under
an expected-error-probability constraint.
In particular, at each time $t\in\{T\probe+1,\dots,T\}$, the rate
$R_t$ is chosen via:
\begin{align}
\label{eq:inst_objective}
\max_{(R_t,\dots,R_T) \in {\mathcal R}^{T-t+1}} &\quad \sum_{\tau=t}^T R_\tau \\
\label{eq:inst_constraint}
\text{subject to} &\quad
\Ep{\varepsilon(\gamma,R_\tau)\ |\ \hat{\gamma}({\bf I}_{T\probe+1})} \leqslant
e^{-\alpha}\\
&\quad \text{for all}\ \tau=t,\ldots,T , \nonumber
\end{align}
where the expectation in (\ref{eq:inst_constraint}) is taken over
\textr{some posterior distribution
$p(\gamma~|~\hat{\gamma}({\bf I}_{T\probe+1}))$. Let us denote
$\hat{\gamma}_{T\probe+1} \triangleq \hat{\gamma}({\bf I}_{T\probe+1})$ and
the set of possible posterior distributions with
${\mathcal{P}}(\hat{\gamma}_{T\probe+1})$, which, in turn, is decided by the
particular choice of the estimator $\hat{\gamma}(\cdot)$.}

While related, the constraints (\ref{eq:constraint}) and
(\ref{eq:inst_constraint}) have an important difference:
the information contained by ${\bf I}_t$ in (\ref{eq:constraint})
is summarized by the \emph{possibly incomplete} statistic
$\hat{\gamma}({\bf I}_{T\probe+1})$ in (\ref{eq:inst_constraint}).
Consequently, satisfaction of (\ref{eq:inst_constraint}) does not
necessarily guarantee satisfaction of (\ref{eq:constraint}) or vice versa.

Due to the fact that the probing period is limited to the first $T\probe$
packets, $\{ R_{T\probe+1},\ldots ,R_T \}$ does not affect the quality
of future SNR estimates, the rate assignment problem
(\ref{eq:inst_objective})-(\ref{eq:inst_constraint}) decouples, and
the value of $R_t$ satisfying
(\ref{eq:inst_objective})-(\ref{eq:inst_constraint}) reduces to
\begin{equation}
\label{eq:optimal_rate}
R_t^*
~\defn~ \arg\max_{R_t\in\mathcal{R}} R_t \text{~~s.t.~~}
\Ep{\varepsilon(\gamma,R_t)\ |\ \hat{\gamma}({\bf I}_{T\probe+1})}
\leqslant e^{-\alpha}.
\end{equation}
Moreover, \eqref{optimal_rate} implies that
$R_t^*$ is invariant to time $t$. Note that the decoupling that occurs here is reminiscent of the decoupling that occurred with ideal rate selection (\ref{eq:objective})-(\ref{eq:constraint}) under perfect channel state information, i.e., (\ref{eq:genie}).

In the next section, we shall see that the choice of estimator plays a key
role in the overall performance of the practical rate adaptation scheme.
Recall that the estimator determines
$p(\gamma\ |\ \hat{\gamma}({\bf I}_{T\probe+1}))$, which determines the
expected error probability constraint. Under certain scenarios, we shall see
that a solution to (\ref{eq:optimal_rate}) does not exist, i.e., that no rates
within $\mathcal{R}$ satisfy the expected error probability constraint.
\textr{Later, in Section~\ref{sec:estimator}, we develop a non-Bayesian estimator in and show that, with that estimator, the set
${\mathcal{P}}(\hat{\gamma}_{T\probe+1})$ will contain merely the class of
Gaussian distributions, asymptotically as $T\probe \to \infty$, for any set, ${\cal P}$, of prior distributions with finite mean and variance for the SNR.}

%% file: necessary.tex
\section{Rate Adaptation with Imperfect CSI}
\label{sec:necessary_condition}

Before studying the practical rate allocator (\ref{eq:optimal_rate}),
we first consider a particular ``naive'' data-rate allocator, in order to draw 
intuition on how estimation errors affect system performance.
Given SNR estimate $\hat{\gamma}$, generated from a particular unbiased
estimator, the naive allocator assigns the data rate
\begin{equation}
\label{eq:naive_assignment}
R_t^{\text{naive}}(\hat{\gamma})
~\defn~ \arg\max_{R_t \in \mathcal{R}} R_t \text{~~s.t.~~}
\varepsilon(\hat{\gamma},R_t) \leqslant
e^{-\alpha}
\end{equation}
for all $t=T\probe+1,\ldots ,T$.
Due to the lack of expectation in the error-probability constraint of
\eqref{naive_assignment}, the naive rates may
violate the desired expected-error-probability constraint in
(\ref{eq:optimal_rate}). 
This follows from the fact that,
when the posterior distribution $p(\gamma_t|\hat{\gamma}_t)$ is
non-atomic (i.e., $\sigma_{\gamma|\hat{\gamma}}^2>0$),
Jensen's inequality\footnote{
  For unbiased $\hat{\gamma}$, (\ref{eq:Jensen_1}) immediately follows from
  Jensen's inequality.  For biased $\hat{\gamma}$, (\ref{eq:Jensen_1}) still
  holds but requires some effort to derive.
  We skip these details since our focus is on unbiased $\hat{\gamma}$.}
implies that
\begin{equation}
\label{eq:Jensen_1}
\Ep{\varepsilon(\gamma,R_t)\ |\ \hat{\gamma}} >
\varepsilon(\hat{\gamma},R_t) ~~\forall R_t.
\end{equation}
Therefore, to ensure the \emph{expected}-error-probability constraint in
\eqref{optimal_rate}, the practical allocator must ``back-off'' the rate
relative to $R_t^{\text{naive}}(\hat{\gamma})$.
To do so, it chooses $R_t^*(\hat{\gamma}) \leqslant
R_t^{\text{naive}}(\hat{\gamma})$, where equality occurs if and only if
the estimation error $N\defn\gamma-\hat{\gamma}$ is zero-valued
(with probability one).

When the estimator is perfect (i.e., $\hat{\gamma}=\gamma$), we note that
the naive rate coincides with the ideal rate under perfect CSI
(i.e., $R_t^{\text{naive}}(\hat{\gamma})=
R_t^{\text{perf-CSI}}(\gamma) |_{\gamma=\hat{\gamma}}$).
In this case, $R_t^{\text{naive}}$ acts as an
upper bound on the ideal $R_t$ under ACK/NAK feedback, as
specified by (\ref{eq:objective})-(\ref{eq:constraint}).
Accordingly, we make the following two definitions.
\begin{definition}
\label{def:rate_penalty}
The \emph{rate penalty} associated with estimator $\hat{\gamma}$
is the smallest $\delta$ (in bits/symbol) that satisfies
\begin{equation}
\label{eq:rate_penalty}
\Ep{\varepsilon(\gamma,R_t^{\text{naive}}(\hat{\gamma})-\delta)\
|\ \hat{\gamma}} \leqslant e^{-\alpha} .
\end{equation}
\end{definition}
\begin{definition}
\label{def:power_penalty}
The \emph{power penalty} associated with estimator $\hat{\gamma}$ is
the smallest scale factor $\mu$ that satisfies
\begin{equation}
\label{eq:power_penalty}
\Ep{\varepsilon(\mu\gamma,R_t^{\text{naive}}(\hat{\gamma}))\
|\ \hat{\gamma}} \leqslant e^{-\alpha} .
\end{equation}
\end{definition}

Next, we analyze two different scenarios for the described rate
adaptation system. In the first scenario, the $n$ symbols in the
packet are assumed to be uncoded QAM symbols, while in the second scenario,
the $n$ symbols are a Gaussian random coded ensemble.  Within the
second scenario, we focus on the high-SNR and low-SNR cases separately.
For both scenarios, we use the analysis presented next, in \secref{approximation}.

\subsection{Gaussian Approximation of the Estimation Error}
\label{sec:approximation}

Under the posterior distribution $p(\gamma|\hat{\gamma})$, let
the estimation error $N=\gamma-\hat{\gamma}$ have the distribution
$q(N|\hat{\gamma})=p(N+\hat{\gamma}|\hat{\gamma})$.
Let $g_{N|\hat{\gamma}}(r)$ and $\Lambda_{N|\hat{\gamma}}(r)$ denote
the \emph{moment generating function} and the \emph{semi-invariant log moment
generating function}~\cite{Gallager:Book:96} of $N$ given $\hat{\gamma}$, respectively.
We assume that there exists some $r_{\max}>0$ such that
$\Lambda_{N|\hat{\gamma}}(r)<\infty$ for all $|r| < r_{\max}$.
It is well known~\cite{Gallager:Book:96} that
$\Lambda_{N|\hat{\gamma}}(0)=0$,
$\Lambda_{N|\hat{\gamma}}'(0)=\Eq{N|\hat{\gamma}}$, and
$\Lambda_{N|\hat{\gamma}}''(0)=\sigma_{N|\hat{\gamma}}^2$.
Then, for any $|r| < r_{\max}$,
\begin{align}
\textstyle
\Eq{\exp(rN)\ |\ \hat{\gamma}} &= g_{N|\hat{\gamma}}(r) =
\exp\big(\Lambda_{N|\hat{\gamma}}(r)\big) \\
\label{eq:taylor_thm1}
&= \exp \Big( \Eq{N\ |\ \hat{\gamma}}r + \frac{1}{2}
\Lambda_{N|\hat{\gamma}}''(r')r^2 \Big)
\end{align}
for some $r'$ between $0$ and $r$ (having the same sign as $r$),
where~(\ref{eq:taylor_thm1}) follows from Taylor's theorem.
Furthermore, applying Taylor's theorem to the third-order expansion, we get
\begin{equation}
\label{eq:taylor_thm2}
\textstyle
g_{N|\hat{\gamma}}(r)=\exp\Big( \Eq{N\
|\ \hat{\gamma}} r + \frac{1}{2}\sigma_{N|\hat{\gamma}}^2 r^2
+ \frac{1}{6}\Lambda_{N|\hat{\gamma}}'''(r'')r^3 \Big)
\end{equation}
for some $r''$ between $0$ and $r$.

In many cases, the first two terms of the expansion~(\ref{eq:taylor_thm2})
lead to insightful expressions to illustrate the impact of the first-
and second-order statistics of ``channel variability.'' This will be
referred to as the \emph{Gaussian approximation}, since, when $N|\hat{\gamma}$
is Gaussian, the cumulants of higher order than the variance vanish.


Further, for an unbiased estimator, $\Eq{N\ |\ \hat{\gamma}}=0$.
In this case, the Gaussian approximation yields the simple
second-order approximation:
\begin{equation}
\label{eq:second_order}
\Lambda_{N|\hat{\gamma}}(r) \approx
\frac{1}{2}\sigma_{N|\hat{\gamma}}^2 r^2 .
\end{equation}
Regardless of the posterior distribution $p(N|\hat{\gamma})$, the
approximation (\ref{eq:second_order}) is {\bf asymptotically accurate}
for the non-Bayesian estimator proposed in Section~\ref{sec:estimator}, which 
is asymptotically unbiased and asymptotically normal, as will be proved.



\subsection{Rate Adaptation with Uncoded QAM}
\label{sec:uncoded_QAM}


Here, we study the scenario in which the $n$ symbols $\{X_{t,k}\}_{k=1}^n$
of packet $t$ are uncoded and selected from a QAM constellation
of size $M_t$. Since the constellation size is constant over the packet,
the rate equals $R_t = \log_2{M_t}$ bits/symbol.
The following is a tight\footnote{The bound holds within approximately 1 dB
from the true value for a wide range of SNRs \cite[p.~289]{Goldsmith:Book:05}.}
approximation~\cite[p.~289]{Goldsmith:Book:05} on
the {\em symbol error rate} associated with minimum-distance
decision making~\cite[p.~280]{Proakis:Book:95}:
\begin{equation}
\label{eq:sym_error_uncoded} \varepsilon_{k}(\gamma,R_t) \approx
0.2\exp\left(-\frac{3}{2}\frac{\gamma}{2^{R_t}-1} \right) .
\end{equation}
The associated packet error rate is
\begin{equation}
\label{eq:packet_err}
\varepsilon(\gamma,R_t)=1-(1-\varepsilon_{k}(\gamma,R_t))^n ,
\end{equation}
since $\varepsilon_{k}(\gamma,R_t)$ remains constant for all
$k$, as $\gamma$ and $R_t$ remain constant over the packet.

Since we can write
\begin{equation}
\label{eq:error_approximation}
(1-\varepsilon_{k}(\gamma,R_t))^n \leqslant
1-n\varepsilon_{k}(\gamma,R_t) +\frac{1}{2}n(n-1)
\varepsilon_{k}^2(\gamma,R_t) ,
\end{equation}
it follows that $\varepsilon(\gamma,R_t) >
\frac{n}{2}\varepsilon_{k}(\gamma,R_t)$ for all $(\gamma,R_t)$ such that
$\varepsilon_{k}(\gamma,R_t)<\frac{1}{n-1}$.
Similarly, \eqref{packet_err} implies that
$\varepsilon(\gamma,R_t)<1-(1-\frac{1}{n-1})^n$
for the same $(\gamma,R_t)$.
This latter bound is an increasing function of $n$, and, for $n \gg 1$,
it approximately equals $1-e^{-1}$, which is much higher than typical
error rates.
We assume that $n$ is large enough and the possible outcomes of
$(\gamma,R_t)$ are such that
$\varepsilon(\gamma,R_t) > \frac{n}{2}\varepsilon_{k}(\gamma,R_t)$ for
all $t$ with probability close to $1$.
We further elaborate on this next, after we derive a sufficient condition for
the error constraint to be met.

To meet the expected-error-probability constraint (\ref{eq:optimal_rate}), it is necessary that
\begin{align}
\nonumber \lefteqn{ \frac{n}{2} \Ep{\varepsilon_{k}(\gamma,R_t)\ |\ \hat{\gamma}} }\\
&\approx \frac{n}{2} \Ep{\left. 0.2\exp\left(
-\frac{3}{2}\frac{\gamma}{2^{R_t}-1} \right) \ \right|\ \hat{\gamma}} \\
\label{eq:cond_uncoded_QAM}
&= \frac{n}{2} \Eq{\left. 0.2\exp\left( -\frac{3}{2}\frac{\hat{\gamma}+N}{2^{R_t}-1} \right) \ \right|\
\hat{\gamma}} \leqslant e^{-\alpha}.
\end{align}
Using the unbiased Gaussian approximation (\ref{eq:second_order}), condition (\ref{eq:cond_uncoded_QAM}) can be rewritten as follows,
after taking the natural log of both sides:
\begin{equation}
\label{eq:eq_uncoded_QAM}
-\frac{3}{2}\frac{\hat{\gamma}}{2^{R_t}-1}
+\frac{ \sigma_{N|\hat{\gamma}}^2 }{2}
\left( \frac{3}{2} \frac{1}{2^{R_t}-1} \right)^2
\leqslant -\alpha-\ln 0.1n .
\end{equation}
For the existence of a feasible rate $R_t$, the solution set for
Inequality \eqref{eq_uncoded_QAM} must be non-empty, for which it is necessary
that
\begin{equation}
\label{eq:condition_uncoded_QAM}
\frac{\hat{\gamma}^2}{\sigma_{N|\hat{\gamma}}^2} \geqslant
2(\alpha+\ln 0.1n) .
\end{equation}
Condition~(\ref{eq:condition_uncoded_QAM}) implies that
$\hat{\gamma}^2/\sigma_{N|\hat{\gamma}}^2$,
the \emph{effective SNR of estimator $\hat{\gamma}$},
must be at least $2(\alpha+\ln 0.1n)$ to guarantee an
expected error rate of $e^{-\alpha}$.
Using similar steps,\footnote{
  From (\ref{eq:packet_err}) and the fact that
  $(1-\epsilon_{t,k})^n > 1-n\epsilon_{t,k}$, we have
  $\varepsilon(\gamma,R_t)\leqslant n \varepsilon_{k}(\gamma,R_t)$ for
  all $(t,k)$ with probability $1$.
  Consequently, for satisfaction of (\ref{eq:optimal_rate}),
  it is sufficient that
  $n\Ep{\varepsilon_{k}(\gamma,R_t)\ |\ \hat{\gamma}} \leqslant e^{-\alpha}$.
  Replicating (\ref{eq:cond_uncoded_QAM})-(\ref{eq:condition_uncoded_QAM}),
  we obtain the sufficiency condition.
  }
a sufficient condition
$\hat{\gamma}^2/\sigma_{N|\hat{\gamma}}^2 \geqslant 2(\alpha+\ln 0.2n)$
can also be derived, illustrating the tightness of
(\ref{eq:condition_uncoded_QAM}).
We will investigate the difficulty of achieving this condition in the
next section.

Given that (\ref{eq:condition_uncoded_QAM}) is satisfied, one can solve
(\ref{eq:eq_uncoded_QAM}) to find the upper bound
$R_t^* \leqslant \bar{R}_t^*(\hat{\gamma},\sigma_{N|\hat{\gamma}}^2)$, where
\begin{align}
\label{eq:rate_bnd_QAM}
\bar{R}_t^*(\hat{\gamma},\sigma_{N|\hat{\gamma}}^2) 
\defn& \log_2\left( 1+\hat{\gamma}\cdot
\frac{3}{2}\frac{\sigma_{N|\hat{\gamma}}^2}{\hat{\gamma}^2}
\Bigg( 1
\right.
\nonumber\\&
\left.
\mbox{}
-\sqrt{1-2(\alpha+\ln 0.1n)
\frac{\sigma_{N|\hat{\gamma}}^2}{\hat{\gamma}^2}
}\Bigg)^{-1}
\right).
\end{align}
Fig.~\ref{fig:rate_penalty_QAM} plots the upper bound (\ref{eq:rate_bnd_QAM})
as a function of the estimator's effective SNR
$\hat{\gamma}^2/\sigma_{N|\hat{\gamma}}^2$ for $\hat{\gamma}\in\{13,20,25\}$ dB,
a desired packet error rate of $e^{-\alpha} = 10^{-3}$,
and a packet size of $n=500$ symbols.
The naive rate allocation
\begin{equation}
R_t^{\text{naive}}(\hat{\gamma}) =\log_2\left( 1+\hat{\gamma}\cdot\frac{3}{2}
\frac{1}{\alpha+\ln 0.1n} \right)
\end{equation}
(derived from \eqref{cond_uncoded_QAM} with $N=0$)
is also shown on the same plot.
The required effective SNR $\hat{\gamma}^2/\sigma_{N|\hat{\gamma}}^2$,
as imposed by (\ref{eq:condition_uncoded_QAM}), is $21.6$ here.
Fig.~\ref{fig:rate_penalty_QAM} shows that
$R_t^{\text{naive}}(\hat{\gamma})< 2$ bits/symbol for
$\hat{\gamma} \leqslant 13$ dB.
Since $2$ bits/symbol is the minimum possible rate for uncoded QAM,
we conclude that it is impossible to meet the target packet-error rate
of $10^{-3}$ when $\hat{\gamma}\leqslant 13$ dB, even with perfect CSI.

\begin{figure*}[th!]
  \begin{center}
    \subfigure{
      \psfrag{gam=25}[B][B][0.7]{\sf $\hat{\gamma}=25$ dB}
      \psfrag{gam=20}[B][B][0.7]{\sf $\hat{\gamma}=20$ dB}
      \psfrag{gam=13}[B][B][0.7]{\sf $\hat{\gamma}=13$ dB}
      \psfrag{Rt-s}[l][l][0.65]{$\bar{R}_t^*$}
      \psfrag{Rt-n}[l][l][0.65]{$R_t^{\text{naive}}$}
      \psfrag{effective SNR}[t][][0.8]{
            $\hat{\gamma}^2/\sigma_{N|\hat{\gamma}}^2$}
      \label{fig:rate_penalty_QAM}
      \includegraphics[width=3.0in]{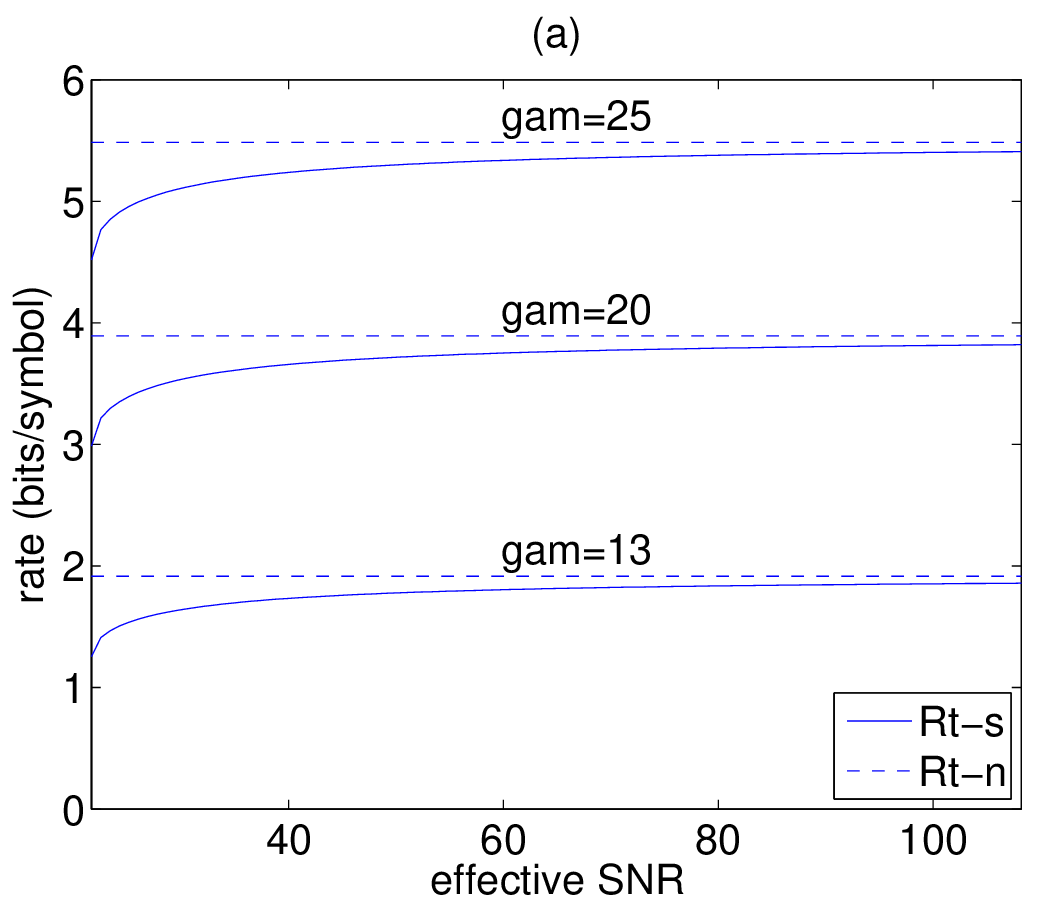}
    }
    \subfigure{
      \psfrag{effective SNR}[t][][0.8]{
            $\hat{\gamma}^2/\sigma_{N|\hat{\gamma}}^2$}
      \label{fig:power_penalty_QAM}
      \includegraphics[width=3.07in]{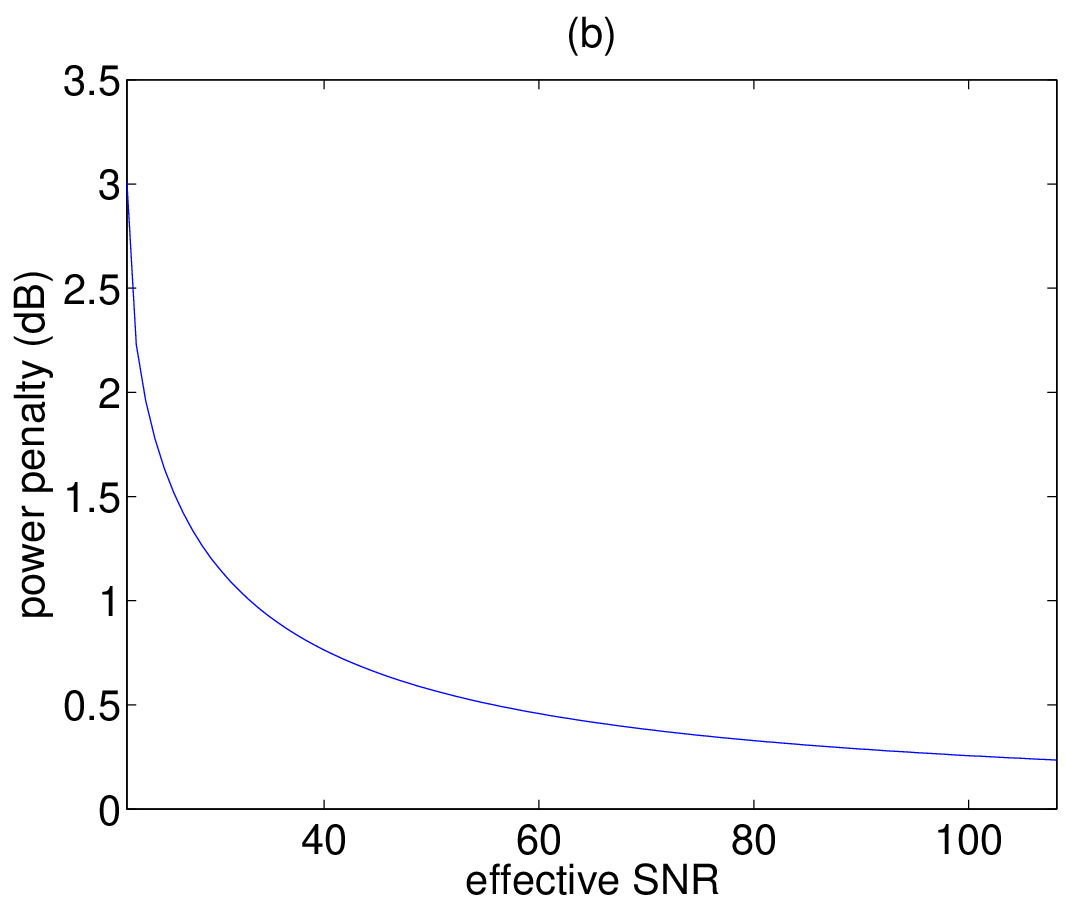}
    }
  \caption{For QAM signaling,
    (a) rates $\bar{R}_t^*$ and $R_t^{\text{naive}}$ versus
    estimator's effective SNR
    $\hat{\gamma}^2/\sigma_{N|\hat{\gamma}}^2$, and
    (b) power penalty lower bound $\underline{\mu}$
    versus estimator's effective SNR
    $\hat{\gamma}^2/\sigma_{N|\hat{\gamma}}^2$.}
  \end{center}
\end{figure*}

By definition, the rate penalty is the smallest $\delta$ that satisfies
$\delta = R_t^{\text{naive}}(\hat{\gamma})
- R_t^*(\hat{\gamma},\sigma_{N|\hat{\gamma}}^2)$.
Thus, an upper bound on $\delta$ is given by
\begin{equation}
\bar{\delta}(\hat{\gamma},\sigma_{N|\hat{\gamma}}^2) \defn
R_t^{\text{naive}}(\hat{\gamma})
- \bar{R}_t^*(\hat{\gamma},\sigma_{N|\hat{\gamma}}^2).
\end{equation}
From \figref{rate_penalty_QAM}, we can see that
$\bar{\delta}(\hat{\gamma},\sigma_{N|\hat{\gamma}}^2)$
depends on the effective SNR $\hat{\gamma}^2/\sigma_{N|\hat{\gamma}}^2$:
it is significant when the effective SNR is near the minimum value
established by (\ref{eq:condition_uncoded_QAM}), but shrinks as
$\hat{\gamma}^2/\sigma_{N|\hat{\gamma}}^2$ gets large.
In addition, $\bar{\delta}(\hat{\gamma},\sigma_{N|\hat{\gamma}}^2)$
grows in proportion to $\hat{\gamma}$.

By definition, the power penalty is the smallest $\mu$ that
satisfies $R_t^*(\hat{\gamma}) = R_t^{\text{naive}}(\hat{\gamma}/\mu)$.
Thus, a lower bound $\underline{\mu}(\hat{\gamma},\sigma_{N|\hat{\gamma}}^2)$
on the power penalty can be found by solving
$\bar{R}_t^*(\hat{\gamma},\sigma_{N|\hat{\gamma}}^2)
= R_t^{\text{naive}}(\hat{\gamma}/\mu)$ for $\mu$.
The power penalty lower bound
$\underline{\mu}(\hat{\gamma},\sigma_{N|\hat{\gamma}}^2)$
is plotted in Fig.~\ref{fig:power_penalty_QAM} as a function of effective SNR
$\hat{\gamma}^2/\sigma_{N|\hat{\gamma}}^2$
for the same expected packet-error rate, $10^{-3}$, and packet size, $n=500$,
as in Fig.~\ref{fig:rate_penalty_QAM}.
The power penalty is seen to be as high as $3$ dB when the effective SNR is
near the minimum value established by (\ref{eq:condition_uncoded_QAM}), but
shrinks as $\hat{\gamma}^2/\sigma_{N|\hat{\gamma}}^2$ gets large.

\subsection{Rate Adaptation with Random Gaussian Ensembles}
\label{sec:coded_gaussian}

Next, we study the random coding~\cite{Gallager:Book:68,Viterbi:Book:79}
scenario in which the codewords are selected from a Gaussian ensemble. Let
$R_{\max}$ be the maximum rate in $\mathcal{R}$.
Then the Gaussian ensemble consists of $2^{nR_{\max}}$ possible packets,
where each symbol, $X_{t,k}$, of packet $t$ is chosen independently from a
${\mathcal N}(0,1)$ distribution.\footnote{
 We use real-valued symbols, instead of complex-valued symbols, for
 simplicity.  Consequently, the data rates will be represented in units of bits per real-symbol.  For fair comparison with uncoded QAM, one should simply double these data rates.}
(We use unit variance here because earlier we assumed $\E{|X_{t,k}|^2}=1$.)
At time $t$, say that transmission rate $R_t\in\mathcal{R}$ is chosen.
Then one packet from a size-$2^{nR_t}$ subset of the initially generated
set of $2^{nR_{\max}}$ packets is chosen arbitrarily for transmission.

The receiver is assumed to know the subsets of possible packets corresponding to each admissible rate $R_t \in \mathcal{R}$.
Based on its observation of the $t^{th}$ packet, the receiver finds the most
likely packet within the subset of $2^{nR_t}$ possible packets.
Note that, unlike the uncoded QAM scenario, where each symbol is decoded
separately, here the entire packet is decoded as a unit.
An upper bound for the associated decoding error probability is
(e.g., \cite{Gallager:Book:68})
\begin{equation}
\label{eq:error_bnd_random}
\varepsilon(\gamma,R_t) \leqslant \exp \left( n\rho \left[ R_t\ln 2 - \frac{1}{2}\ln \left( 1+\frac{\gamma}{1+\rho} \right) \right] \right),
\end{equation}
where $\rho\in[0,1]$ is the union bound parameter. One can minimize (\ref{eq:error_bnd_random}) over $\rho\in[0,1]$ to find the tightest bound, if so desired. To satisfy the expected-error-probability constraint (\ref{eq:optimal_rate}), it suffices that there exists a $\rho\in[0,1]$ for which
\begin{equation}
\label{eq:condition_coded_random}
\Ep{\left. \exp \left( n\rho \left[ R_t\ln 2 - \frac{1}{2}\ln \left(
1+\frac{\gamma}{1+\rho} \right) \right] \right)\ \right| \ \hat{\gamma}}
\leqslant e^{-\alpha}.
\end{equation}

\subsubsection{Low-SNR Regime}
When $\Prob{\gamma \ll 1\ |\ \hat{\gamma}} \approx 1$, we can write
\begin{equation}
\label{eq:low_SNR}
\ln \left( 1+\frac{\gamma}{1+\rho} \right) \approx \frac{\gamma}{1+\rho}
=\frac{\hat{\gamma}+N}{1+\rho} .
\end{equation}
For an unbiased estimator, $\Eqsmall{\frac{\hat{\gamma}+N}{1+\rho}
\big| \hat{\gamma}} = \frac{\hat{\gamma}}{1+\rho}$ and
$\varqsmall{\frac{\hat{\gamma}+N}{1+\rho} \big| \hat{\gamma}} =
\frac{\sigma_{N|\hat{\gamma}}^2}{(1+\rho)^2}$.
Thus, using the Gaussian approximation (\ref{eq:second_order}), the constraint (\ref{eq:condition_coded_random}) is satisfied if there exists a $\rho\in[0,1]$ for which
\begin{equation}
\label{eq:cond_low_SNR}
\alpha \leqslant -n\rho\left(R_t \ln 2 - \frac{\hat{\gamma}}{2(1+\rho)} +
\frac{1}{8}\frac{n\rho}{(1+\rho)^2}\sigma_{N|\hat{\gamma}}^2 \right) ,
\end{equation}
or, equivalently, for which
\begin{equation}
\label{eq:rate_ineq_low_SNR}
R_t \leqslant \frac{1}{\ln 2} \left( -\frac{\alpha}{n\rho} +
\frac{\hat{\gamma}}{2(1+\rho)}
-\frac{1}{8}\frac{n\rho}{(1+\rho)^2}\sigma_{N|\hat{\gamma}}^2 \right) .
\end{equation}
Thus, if there exists some $\rho \in [0,1]$ for which the right side
of \eqref{rate_ineq_low_SNR} is positive, then any $R_t$ below it is
feasible. For this to be possible, we need
\[ 2\alpha(1+\rho)^2 - \hat{\gamma} n\rho(1+\rho) +
\frac{1}{4}(n\rho)^2\sigma_{N|\hat{\gamma}}^2 \leqslant 0 \]
for some $\rho\in[0,1]$, which leads to the following necessary
condition\footnote{
  Note that condition (\ref{eq:est_cond_low_SNR}) is not exactly
  analogous to condition~(\ref{eq:condition_uncoded_QAM}). Condition
  (\ref{eq:est_cond_low_SNR}) is necessary for a non-empty solution set to exist for inequality (\ref{eq:rate_ineq_low_SNR}), whereas
  (\ref{eq:condition_uncoded_QAM}) is necessary for the existence of a feasible rate that satisfies the expected-error bound. In order to derive an analogous necessary condition, one can use a sphere-packing (SP) bound for the Gaussian channel (see, e.g., \cite{Wiechman:TIT:08}). With the SP lower bound, our findings would be qualitatively similar, but the derivation would be extremely tedious.  For this reason, we assume that the upper bound is a good approximation for the actual error rate.}
for the estimator:
\begin{equation}
\label{eq:est_cond_low_SNR}
\frac{\hat{\gamma}^2}{\sigma_{N|\hat{\gamma}}^2} \geqslant 2\alpha .
\end{equation}
One can then find
an upper bound on $R_t\in\mathcal{R}$ satisfying
\eqref{condition_coded_random} as follows:
\begin{align}
\label{eq:rate_bnd_low_SNR}
\lefteqn{\bar{R}^*_t(\hat{\gamma},\sigma_{N|\hat{\gamma}}^2)}\nonumber\\
&= \max_{\rho\in[0,1]} \frac{1}{\ln 2} \left( -\frac{\alpha}{n\rho} +
\frac{\hat{\gamma}}{2(1+\rho)}
-\frac{1}{8}\frac{n\rho}{(1+\rho)^2}\sigma_{N|\hat{\gamma}}^2 \right) .
\end{align}
Likewise, one can deduce from \eqref{error_bnd_random} and
\eqref{low_SNR} that the naive rate is
\begin{equation}
\label{eq:naive_low_SNR} R_t^{\text{naive}}(\hat{\gamma}) = \max_{\rho \in[0,1]}
\frac{1}{\ln 2}\left( -\frac{\alpha}{n\rho} +
\frac{\hat{\gamma}}{2(1+\rho)} \right) .
\end{equation}

The rate upper bound $\bar{R}^*_t$ is
plotted in Fig.~\ref{fig:rate_penalty_low_SNR} as a function of the estimator's
effective SNR $\hat{\gamma}^2/\sigma_{N|\hat{\gamma}}^2$
for $\hat{\gamma}\in\{-3,-8, -12\}$ dB, a desired packet error rate of $e^{-\alpha}=10^{-3}$, and a packet size of $n=500$ symbols.
The rate $R_t^{\text{naive}}$ from (\ref{eq:naive_low_SNR}) is
also shown on the same plot.
Every point on the rate curves was computed using the optimal
value of $\rho\in[0,1]$, found numerically.
We note that, with these parameters, (\ref{eq:est_cond_low_SNR}) implies that
$\hat{\gamma}^2/\sigma_{N|\hat{\gamma}}^2$ must be at least $13.8$.
Figure~\ref{fig:rate_penalty_low_SNR} also shows that the rate penalty
$\bar{\delta}(\hat{\gamma},\sigma_{N|\hat{\gamma}}^2)=
R_t^{\text{naive}}(\hat{\gamma})
-\bar{R}_t^*(\hat{\gamma},\sigma_{N|\hat{\gamma}}^2)$ is significant when
$\hat{\gamma}^2/\sigma_{N|\hat{\gamma}}^2$ is near the lower bound
established by (\ref{eq:est_cond_low_SNR}), but that the rate penalty
shrinks as $\hat{\gamma}^2/\sigma_{N|\hat{\gamma}}^2$ increases.

\begin{figure*}[th!]
  \begin{center}
   \subfigure{
      \psfrag{gam=-3}[B][B][0.7]{\sf $\hat{\gamma}=-3$ dB}
      \psfrag{gam=-8}[B][B][0.7]{\sf $\hat{\gamma}=-8$ dB}
      \psfrag{gam=-12}[B][B][0.7]{\sf $\hat{\gamma}=-12$ dB}
      \psfrag{Rt-s}[l][l][0.65]{$\bar{R}_t^*$}
      \psfrag{Rt-n}[l][l][0.65]{$R_t^{\text{naive}}$}
      \psfrag{effective SNR}[t][][0.8]{
            $\hat{\gamma}^2/\sigma_{N|\hat{\gamma}}^2$}
     \label{fig:rate_penalty_low_SNR}
     \includegraphics[width=3.0in]{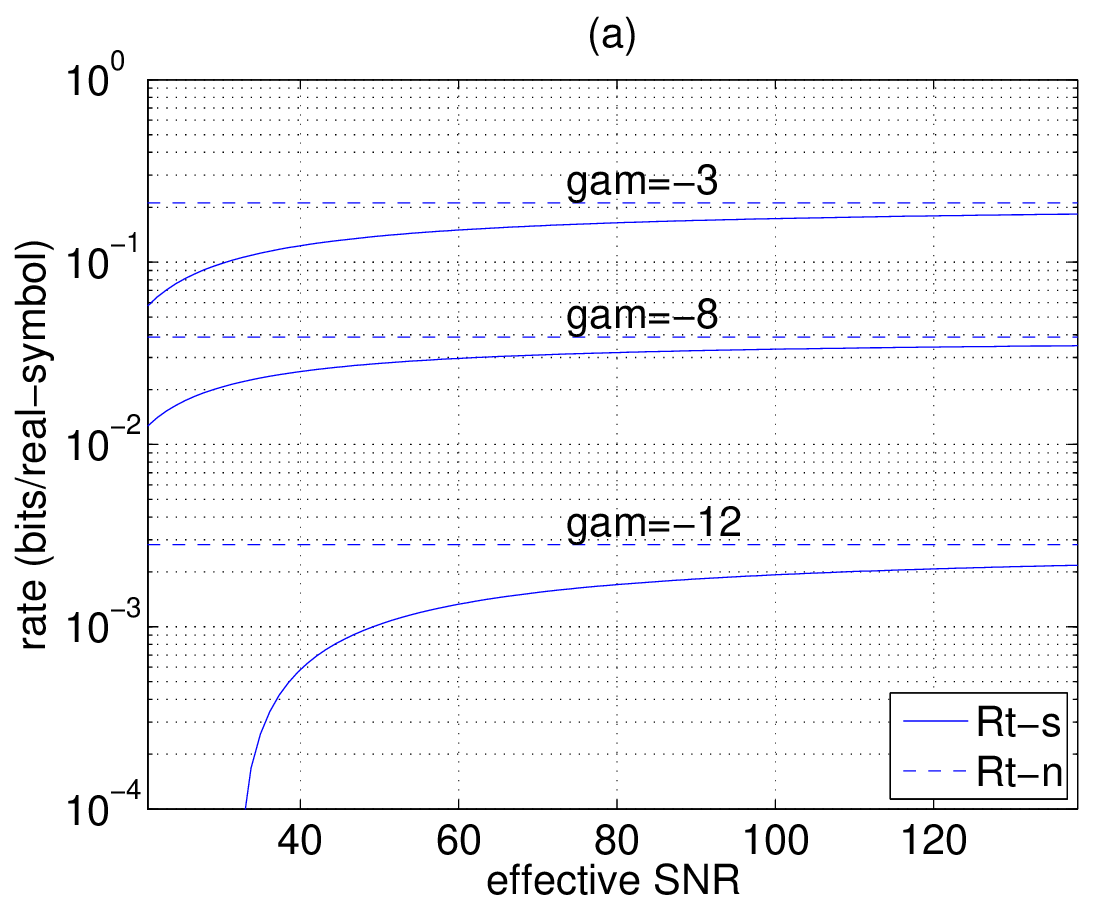}
   }
   \subfigure{
     \psfrag{Rt-n}[l][l][0.65]{$R_t^{\text{naive}}$}
      \psfrag{Rt-s eff-SNR=60}[l][l][0.65]{\sf $\bar{R}_t^*$
            at $\hat{\gamma}^2/\sigma_{N|\hat{\gamma}}^2=60$}
      \psfrag{Rt-s eff-SNR=100}[l][l][0.65]{\sf $\bar{R}_t^*$
            at $\hat{\gamma}^2/\sigma_{N|\hat{\gamma}}^2=100$}
      \psfrag{Shannon limit}[l][l][0.6]{\sf \,naive Shannon limit}
      \psfrag{SNR estimate (dB)}[t][][0.8]{\sf $\hat{\gamma}$ (dB)}
      \label{fig:power_penalty_low_SNR}
      \includegraphics[width=3.0in]{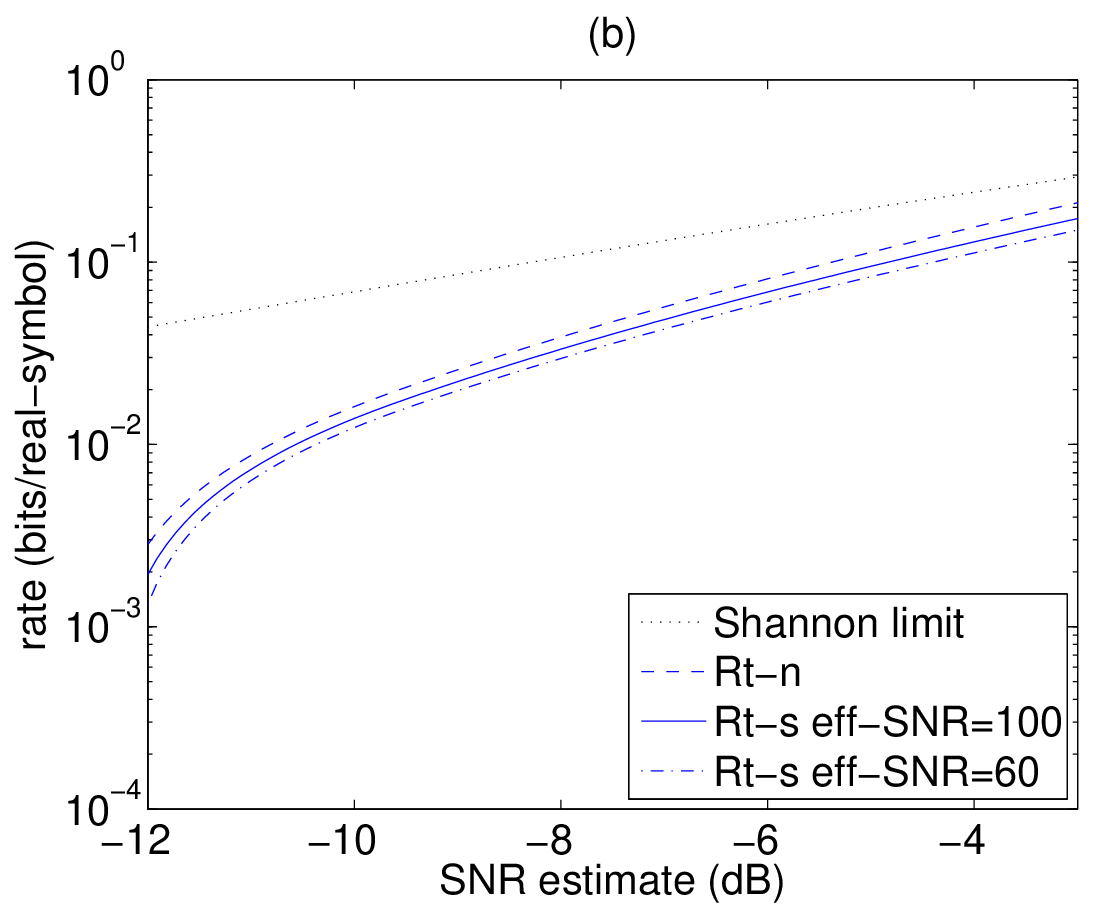}
   }
  \end{center}
  \caption{For Gaussian signaling at low SNR,
    rates $\bar{R}^*_t$ and $R_t^{\text{naive}}$ versus
    (a) estimator's effective SNR
    $\hat{\gamma}^2/\sigma_{N|\hat{\gamma}}^2$ and
    (b) estimated SNR $\hat{\gamma}$.}
\end{figure*}

For the same target packet error rate ($10^{-3}$) and packet size ($n=500$),
Fig.~\ref{fig:power_penalty_low_SNR} plots $\bar{R}^*_t$
versus $\hat{\gamma}$ for estimator effective SNR
$\hat{\gamma}^2/\sigma_{N|\hat{\gamma}}^2\in\{60,100\}$.
In the same figure, $R_t^{\text{naive}}$ and the
``naive'' Shannon limit (i.e., ergodic capacity)
$\frac{1}{2}\log_2(1+\hat{\gamma})$
bits/real-symbol are shown.
By comparing the naive Shannon limit with $R_t^{\text{naive}}$, one
can observe that, in the low-SNR regime, the power penalty of
Gaussian signaling scheme can be significant, especially at small
values of $\hat{\gamma}$.
From the same plot, one can observe that the additional power penalty due
to imperfect SNR estimation, $\mu(\hat{\gamma})$, is quite small:
less than $0.5$ dB when $\hat{\gamma}^2/\sigma_{N|\hat{\gamma}}^2=100$
and less than $1$ dB when $\hat{\gamma}^2/\sigma_{N|\hat{\gamma}}^2=60$.

\subsubsection{High-SNR Regime}
When $\Prob{\gamma \gg 1\ |\ \hat{\gamma}} \approx 1$, we can write
\begin{align}
\label{eq:high_SNR}
\ln \left( 1+\frac{\gamma}{1+\rho} \right) 
&\approx \ln \left( \frac{\gamma}{1+\rho} \right) 
= \ln \left( \frac{\hat{\gamma}+N}{1+\rho} \right) \nonumber \\
&= \ln \left( \frac{\hat{\gamma}}{1+\rho} \right) + \ln \left(1+
\frac{N}{\hat{\gamma}}\right).
\end{align}
Thus, for an unbiased estimator,
$\Eqsmall{\ln \frac{\gamma}{1+\rho} \big| \hat{\gamma}}\approx
\ln \frac{\hat{\gamma}}{1+\rho}$
and $\varqsmall{\ln \frac{\gamma}{1+\rho} \big| \hat{\gamma}} \approx \frac{\sigma_{N|\hat{\gamma}}^2}{\hat{\gamma}^2}$.
Similar to the low SNR scenario, we can use the Gaussian approximation
(\ref{eq:second_order}) to claim that (\ref{eq:condition_coded_random})
is satisfied if there exists a $\rho\in[0,1]$ for which
\begin{equation}
\label{eq:cond_high_SNR}
\alpha \geqslant -n\rho\left(R_t \ln 2 - \frac{1}{2}\ln
\left( \frac{\hat{\gamma}}{1+\rho} \right) + \frac{1}{8}
\frac{n\rho}{\hat{\gamma}^2/\sigma_{N|\hat{\gamma}}^2} \right) ,
\end{equation}
or, equivalently,
\begin{equation}
\label{eq:rate_ineq_high_SNR}
R_t \leqslant \frac{1}{\ln 2} \left( -\frac{\alpha}{n\rho} + \frac{1}{2}\ln
\left( \frac{\hat{\gamma}}{1+\rho} \right) - \frac{1}{8}
\frac{n\rho}{\hat{\gamma}^2/\sigma_{N|\hat{\gamma}}^2} \right) .
\end{equation}
Hence, if there exists some $\rho \in [0,1]$ for which the right
side of \eqref{rate_ineq_high_SNR} is positive, then any $R_t$ below
it is feasible.  In the high-SNR regime, we have $\hat{\gamma} \gg
1$ with high probability, and thus there almost always exists some
$\rho\in[0,1]$ for which a feasible $R_t>0$ exists. One can deduce
from this observation that, a principal difference between the
high-SNR and low-SNR regimes is that, in the high-SNR regime, the
expected error probability constraint is satisfied much more easily,
with nearly any SNR estimator. One can then find an upper bound on
$R_t\in\mathcal{R}$ satisfying (\ref{eq:condition_coded_random}) as
follows:
\begin{align}
\lefteqn{\bar{R}_t^*(\hat{\gamma},\sigma_{N|\hat{\gamma}}^2)}\nonumber\\
&=\max_{\rho\in[0,1]} \frac{1}{\ln 2}
\left( -\frac{\alpha}{n\rho} + \frac{1}{2}\ln
\left( \frac{\hat{\gamma}}{1+\rho} \right) - \frac{1}{8}
\frac{n\rho}{\hat{\gamma}^2/\sigma_{N|\hat{\gamma}}^2} \right) .
\label{eq:rate_bnd_high_SNR}
\end{align}
Likewise, one can deduce from \eqref{error_bnd_random} and
\eqref{high_SNR} that the naive rate is
\begin{equation}
\label{eq:naive_high_SNR}
R_t^{\text{naive}}(\hat{\gamma}) = \max_{\rho \in[0,1]} \frac{1}{\ln 2}
\left( -\frac{\alpha}{n\rho} +  \frac{1}{2}\ln
\left( \frac{\hat{\gamma}}{1+\rho} \right) \right) .
\end{equation}

\begin{figure*}[th!]
  \begin{center}
    \subfigure{
      \psfrag{gam=25}[B][B][0.7]{\sf $\hat{\gamma}=25$ dB}
      \psfrag{gam=20}[B][B][0.7]{\sf $\hat{\gamma}=20$ dB}
      \psfrag{gam=13}[B][B][0.7]{\sf $\hat{\gamma}=13$ dB}
      \psfrag{Rt-s}[l][l][0.65]{$\bar{R}_t^*$}
      \psfrag{Rt-n}[l][l][0.65]{$R_t^{\text{naive}}$}
      \psfrag{effective SNR}[t][][0.8]{
            $\hat{\gamma}^2/\sigma_{N|\hat{\gamma}}^2$}
      \label{fig:rate_penalty_high_SNR}
      \includegraphics[width=3.0in]{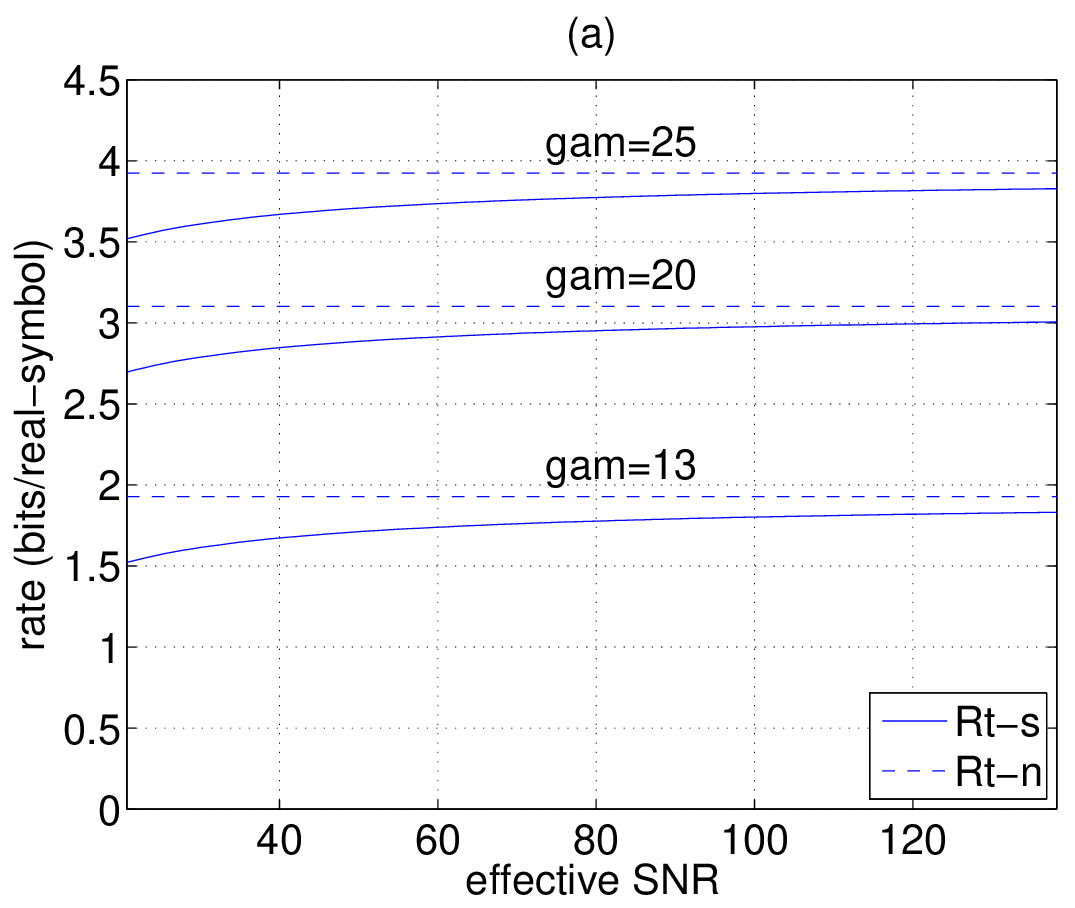}
    }
    \subfigure{
      \psfrag{Rt-n}[l][l][0.65]{$R_t^{\text{naive}}$}
      \psfrag{Rt-s eff-SNR=60}[l][l][0.65]{\sf $\bar{R}_t^*$
            at $\hat{\gamma}^2/\sigma_{N|\hat{\gamma}}^2=60$}
      \psfrag{Rt-s eff-SNR=20}[l][l][0.65]{\sf $\bar{R}_t^*$
            at $\hat{\gamma}^2/\sigma_{N|\hat{\gamma}}^2=20$}
      \psfrag{Shannon limit}[l][l][0.6]{\sf \,naive Shannon limit}
      \psfrag{SNR estimate (dB)}[t][][0.8]{\sf $\hat{\gamma}$ (dB)}
      \label{fig:power_penalty_high_SNR}
      \includegraphics[width=3.05in]{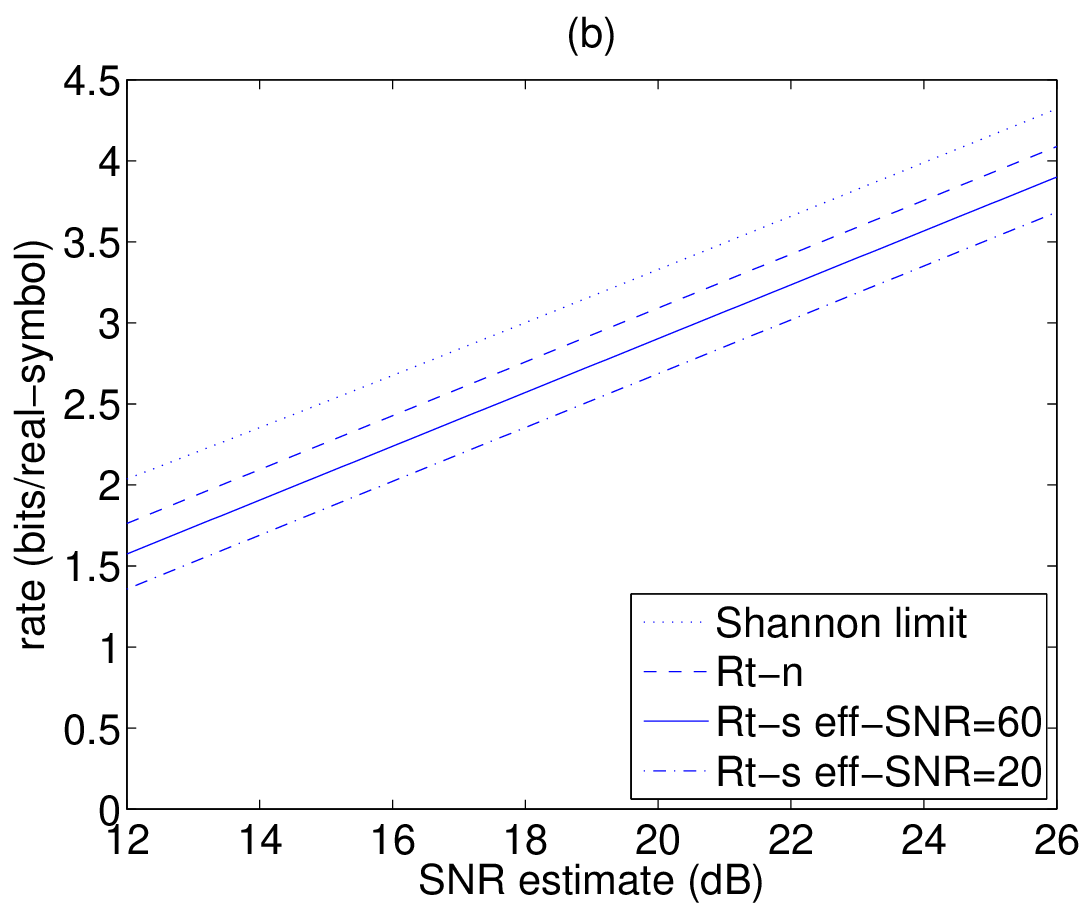}
    }
  \end{center}
  \caption{For Gaussian signaling at high SNR,
    rates $\bar{R}^*_t$ and $R_t^{\text{naive}}$ versus
    (a) estimator's effective SNR
    $\hat{\gamma}^2/\sigma_{N|\hat{\gamma}}^2$ and
    (b) estimated SNR $\hat{\gamma}$.}
\end{figure*}

The rate upper bound $\bar{R}^*_t(\hat{\gamma},\sigma_{N|\hat{\gamma}}^2)$
is plotted in Fig.~\ref{fig:rate_penalty_high_SNR} as a function of the
estimator's effective SNR $\hat{\gamma}^2/\sigma_{N|\hat{\gamma}}^2$
for $\hat{\gamma}\in\{13,20,25\}$ dB, a desired packet error rate of
$e^{-\alpha}=10^{-3}$, and a packet size of $n=500$ symbols.
The rate $R_t^{\text{naive}}(\hat{\gamma})$ from (\ref{eq:naive_high_SNR}) is
also shown on the same plot.
Every point on the rate curves was computed using the optimal
value of $\rho\in[0,1]$, found numerically.
We emphasize that the rates plotted in Fig.~\ref{fig:rate_penalty_high_SNR}
are expressed in bits per \emph{real-symbol}, and thus should be doubled for
fair comparison with the QAM rates presented in
Fig.~\ref{fig:rate_penalty_QAM}.
For Gaussian signaling, if we compare the high-SNR results in
Figs.~\ref{fig:rate_penalty_high_SNR}-\ref{fig:power_penalty_high_SNR}
to the low-SNR results in
Figs.~\ref{fig:rate_penalty_low_SNR}-\ref{fig:power_penalty_low_SNR},
we can see that
the normalized rate penalty $\bar{\delta}/\bar{R}^*_t$ is much
smaller in the high-SNR regime.
For instance, at $\hat{\gamma}^2/\sigma_{N|\hat{\gamma}}^2=20$,
$\bar{\delta}$ is no more than $0.5$ bits/symbol and $\bar{\delta}/\bar{R}^*_t$
is less than 25\% for all three values of $\hat{\gamma}$.
This decrease in rate penalty is expected, since, in the high-SNR regime,
the rate scales roughly with the log of the SNR.

For the same target packet error rate ($10^{-3}$) and packet size ($n=500$),
Fig.~\ref{fig:power_penalty_high_SNR} plots
$\bar{R}^*_t(\hat{\gamma},\sigma_{N|\hat{\gamma}}^2)$
versus $\hat{\gamma}$ for estimator effective SNR
$\hat{\gamma}^2/\sigma_{N|\hat{\gamma}}^2\in\{60,100\}$.
In the same figure, $R_t^{\text{naive}}(\hat{\gamma})$ and the
naive Shannon limit $\frac{1}{2}\log_2(1+\hat{\gamma})$ are shown.
There we observe that, in the high-SNR regime, the power penalty
for Gaussian signaling is constant with $\hat{\gamma}$,
and no more than $1.5$ dB. The additional power penalty
due to imperfect SNR estimation,
$\underline{\mu}(\hat{\gamma},\sigma_{N|\hat{\gamma}}^2)$, is
approximately $1$ dB when $\hat{\gamma}^2/\sigma_{N|\hat{\gamma}}^2=60$ and
approximately $2.5$ dB when $\hat{\gamma}^2/\sigma_{N|\hat{\gamma}}^2=20$.

%% file: capacity.tex
\section{Fundamental Limitations of ACK/NAK-Based Rate Adaptation}
\label{sec:capacity}

In the previous section, we studied the performance of the rate
adaptation system for a generic unbiased estimator. We analyzed the
feasible rates with particular coding/modulation schemes as a
function of the ``quality'' of the estimation provided by the
estimator, for which the relevant metric was the estimator's
effective SNR $\hat{\gamma}^2/\sigma_{N|\hat{\gamma}}^2$.
Note that we assumed no knowledge of the prior SNR distribution $p(\gamma)$.

In this section, we view the SNR of the current block, $\gamma$, as
an unknown parameter,\footnote{We assume that $\gamma$ is a random
  variable, taking on an independent value for each block, but that
  the distribution of $\gamma$ is unknown to the transmitter.}
and pose the estimation of $\gamma$ as a non-Bayesian parameter estimation
problem.
We first investigate the fundamental limitations of SNR estimators
that are based on packet-level ACK/NAK feedback, e.g.,
$\hat{\gamma}=\hat{\gamma}({\bf I}_{T\probe+1})$.
Using that analysis, we show that it is difficult to make good SNR estimates while simultaneously keeping packet-error-rate low. This latter property motivates SNR-estimation via probe packets that come
without error-rate constraints (in contrast to data packets, which are error-rate constrained) as assumed in \secref{model}. Finally, we discuss optimization of the probing period $T\probe$, and we derive
an upper bound on the optimal sum rate $R\tot^*$.

\subsection{Fundamental Limitations of ACK/NAK-Based SNR Estimation}
\label{sec:CR-bound}

Consider the SNR estimator
$\hat{\gamma}(\vec{I}_{T\probe+1})$, based on the $T\probe$ ACK/NAKs
in $${\bf I}_{T\probe+1}
=[F_1,F_2,\ldots,F_{T\probe},R_1,R_2,\ldots ,R_{T\probe}],$$
where $R_t$ denotes the rate and $F_t$ denotes the ACK/NAK feedback
for packet $t$.
In the sequel, we abbreviate $\hat{\gamma}(\vec{I}_{T\probe+1})$ by
$\hat{\gamma}$. Recall that $R_t$ and $F_t$ are connected through the packet error probability $\varepsilon(\gamma,R_t)$, as specified in
(\ref{eq:ack_nack_prob}).

\begin{theorem}
\label{th:CR-bound} For true SNR $\gamma$ and any unbiased estimator
$\hat{\gamma}$ based on $T\probe$ ACK/NAKs, the estimation error
variance, $\sigma_{N|\hat{\gamma}}^2 \defn
\varsmall{\gamma-\hat{\gamma}|\hat{\gamma}}$, is lower bounded by:
\begin{equation}
\label{eq:CR-bound}
\sigma_{N|\hat{\gamma}}^2 \geqslant
\left(\sum_{t=1}^{T\probe} \frac{\left(\varepsilon'(\gamma,R_{t})
\right)^2}{\varepsilon(\gamma,R_{t}) \left[1-\varepsilon
(\gamma,R_{t})\right]}\right)^{-1} ,
\end{equation}
where $\varepsilon(\gamma,R_{t})$ is continuously differentiable
in $\gamma$ and $\varepsilon'(\gamma,R_{t})\defn
\frac{\partial}{\partial \gamma} \varepsilon(\gamma,R_{t})$.
\end{theorem}

\begin{proof}
Given $\gamma$ and the rates $R_1,\ldots,R_{T\probe}$, the feedback
$F_1,\ldots , F_{T\probe}$ satisfies
\begin{align}
\lefteqn{\Prob{F_1=f_1,\ldots ,F_{T\probe}=f_{T\probe}\ |\ \gamma, R_1,\ldots, R_{T\probe}}}\nonumber\\
&= \prod_{t=1}^{T\probe} \Prob{F_t=f_t\ |\ \gamma,R_t}.
\label{eq:observe_1}
\end{align}
Then
\begin{align}
\nonumber
V(\gamma,R_t,f_t) &\defn \frac{\partial}{\partial \gamma} \ln
\Prob{F_t=f_t\ |\ \gamma,R_t} \\
\nonumber
&= \frac{\partial}{\partial \gamma}
\ln \big( [\varepsilon(\gamma,R_t)]^{f_t}
[1-\varepsilon(\gamma,R_t)]^{1-f_t} \big) \\
\label{eq:the_V_fnc}
&= \frac{\varepsilon'(\gamma,R_t)}{1-\varepsilon(\gamma,R_t)}
\left( \frac{f_t}{\varepsilon(\gamma,R_t)} - 1 \right) .
\end{align}
The Fisher information~\cite{Poor:Book:94} associated with $F_t$ is:
\begin{align}
\nonumber
\Phi(\gamma,R_t) &=
\var{ V(\gamma,R_t,f_t)\ |\ \gamma,R_t} \\
\label{eq:fisher_info}
&= \frac{\left( \varepsilon'(\gamma,R_t)
\right)^2}{\varepsilon(\gamma,R_t)
\left[1-\varepsilon(\gamma,R_t)\right]} ,
\end{align}
and the cumulative Fisher information is
$\sum_{t=1}^{T\probe} \Phi(\gamma,R_t)$. Theorem~\ref{th:CR-bound}
follows since the Cramer-Rao lower bound (CRLB) for unbiased
estimators is the reciprocal of the Fisher information~\cite{Poor:Book:94}.
\end{proof}

\subsection{Lower Bounds on the Required Probing Period $T\probe$}
\label{sec:constant_rate_probing}


In \secref{necessary_condition}, we derived lower bounds
(\ref{eq:condition_uncoded_QAM}) and (\ref{eq:est_cond_low_SNR})
on the value of
$\hat{\gamma}^2/\sigma_{N|\hat{\gamma}}^2$ (i.e., the estimator's
effective SNR) required to facilitate the use of data transmission via
uncoded QAM signaling and randomly coded Gaussian signaling, respectively.
In this section, we translate those lower bounds (on required
$\hat{\gamma}^2/\sigma_{N|\hat{\gamma}}^2$)
into lower bounds on required probe-duration $T\probe$, recognizing that the
quality of SNR estimates (and thus $\hat{\gamma}^2/\sigma_{N|\hat{\gamma}}^2$)
increases with $T\probe$.
From these bounds, we shall see that the required value of $T\probe$ depends
heavily on the probe error rate, and in particular that the required value
of $T\probe$ grows very large as the probe error rate decreases.
This motivates the optimization of probe error rate, which requires the
decoupling of probe error rate from data error rate (since the latter is
usually constrained by the application).

In this section, we assume that both the modulation/coding scheme and the
rate is fixed over the probe interval, i.e., that $R_t = R\probe$ for
$t\in\{1,\dots,T\probe\}$.
In this case, the CRLB (\ref{eq:CR-bound}) reduces to
\begin{equation}
\label{eq:CR-bound_fixed_rate}
\sigma_{N|\hat{\gamma}}^2 \geqslant
\frac{1}{T\probe} \frac{\varepsilon(\gamma,R\probe)
\left[1-\varepsilon(\gamma,R\probe)\right]}
{\left(\varepsilon'(\gamma,R\probe)\right)^2} ,
\end{equation}
which is inversely proportional to $T\probe$.

Recall that, to make uncoded QAM signaling feasible, condition
(\ref{eq:condition_uncoded_QAM}) must be satisfied, and to make
random Gaussian signaling feasible in the low-SNR regime, condition
(\ref{eq:est_cond_low_SNR}) must be satisfied.
Though (\ref{eq:condition_uncoded_QAM}) and (\ref{eq:est_cond_low_SNR})
are expressed in terms of the estimator's effective SNR,
we can rewrite them as
$\sigma_{N|\hat{\gamma}}^2
\leqslant \frac{1}{2}\hat{\gamma}^2/(\alpha+\ln 0.1 n)$
and
$\sigma_{N|\hat{\gamma}}^2
\leqslant \frac{1}{2}\hat{\gamma}^2/\alpha$,
respectively, and apply the CRLB (\ref{eq:CR-bound_fixed_rate}) to
arrive (see Appendix~\ref{sec:appendix_1}) at the following.
For uncoded QAM, we need $T\probe \geqslant T\probe^{\min}$, where
\begin{equation}
\label{eq:condition_CR-bound_QAM}
T\probe^{\min} =
\frac{ 2(\alpha+\ln 0.1 n) ~\varepsilon(\gamma,R\probe) }
 { (1-\varepsilon(\gamma,R\probe))
 \big[(1-\varepsilon(\gamma,R\probe))^{-1/n}-1\big]^2
 \ln^2\big(5\big(1-(1-\varepsilon(\gamma,R\probe))^{1/n}\big)\big)
 (n\hat{\gamma}/\gamma)^2 } ,
\end{equation}
and for random Gaussian signaling in the low-SNR regime, we need
$T\probe \geqslant T\probe^{\min}$, where
\begin{equation}
\label{eq:condition_CR-bound_Gaussian}
 T\probe^{\min} =
 \frac{8\alpha \big(1-\varepsilon(\gamma,R\probe)\big)
 \big(1+\rho^*+\gamma\big)^2}
 {\varepsilon(\gamma,R\probe) (n\rho^*\hat{\gamma})^2}
\end{equation}
and where $\rho^*$ is the union bound parameter corresponding to the tightest
error bound (\ref{eq:error_bnd_random}), which itself depends on
$\gamma$, $R\probe$, and $n$.


\begin{figure}[th!]
  \begin{center}
    \psfrag{gam=-3}[B][B][0.6]{\sf $\hat{\gamma}=-3$ dB}
    \psfrag{gam=-7}[B][B][0.6]{\sf $\hat{\gamma}=-7$ dB}
    \psfrag{gam=-10}[B][B][0.6]{\sf $\hat{\gamma}=-10$ dB}
    \psfrag{probe packet-error rate}[][][0.8]{\sf $\varepsilon(\gamma,R\probe)$}
    \psfrag{Tp-min}[][][0.8]{\sf $T\probe^{\min}$}
    \includegraphics[width=3.1in]{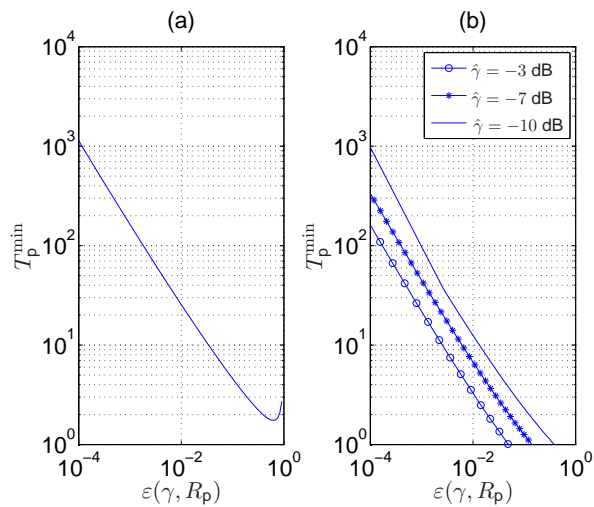}
  \end{center}
  \caption{Lower bound on required probing duration $T\probe^{\min}$ versus
    probe packet-error rate $\varepsilon(\gamma,R\probe)$ for
    (a) uncoded QAM and (b) random Gaussian signaling.}
  \label{fig:cr-bound}
\end{figure}

Figures~\ref{fig:cr-bound}(a)-(b) plot $T\probe^{\min}$ as a function of the
probe error rate $\varepsilon(\gamma,R\probe)$ for uncoded QAM signaling
and random Gaussian signaling, respectively.
For the plots, we assume $\hat{\gamma} \approx \gamma$, which eliminates
the dependence of $T\probe^{\min}$ on $\hat{\gamma}$ and $\gamma$ in the
QAM case; for the Gaussian case, we show $T\probe$ for the values
$\hat{\gamma}\in\{-3,-7,-10\}$ dB.
As in our previous plots, we assumed $n=500$ and $e^{-\alpha}=10^{-3}$.
The key observation to make from these plots is that the number of probe
packets increases quickly as $\varepsilon(\gamma,R\probe)$ shrinks.
In fact, the plots suggest that $T\probe$ is roughly proportional to
$1/\varepsilon(\gamma,R\probe)$. This inverse relationship is
somewhat intuitive because, given a probe packet-error rate of
$\varepsilon(\gamma,R\probe)$, one must wait for
$1/\varepsilon(\gamma,R\probe)$ packets (on average) to see a single
NAK. Recall, however, that \figref{cr-bound} shows only a
\emph{lower bound} $T\probe^{\min}$ on the probe duration required
for communication with positive rate; the optimal value of $T\probe$
is expected to be even larger.

The main conclusion to draw from this section is that, to keep the probing
period small, one must allow relatively high probe error rate
$\varepsilon(\gamma,R\probe)$.
For systems which estimate SNR using only ACK/NAK feedback from data packets,
this implies that if the data error rate $e^{-\alpha}$ is small, then
the number of packets required to get a decent SNR estimate will be large.
Such systems would only be suitable for channels that are very slowly fading.

\subsection{An Upper Bound on the Optimal Sum-Rate }
\label{sec:optimal_rate}

Recall that, in our practical rate adaptation system, the data packet rates
$\{R_t\}_{t=T\probe+1}^{R_T}$ are chosen based on the SNR estimated using
ACK/NAKs from probe packets with rates $\{R_t\}_{t=1}^{T\probe}$.
To complete the system design, we must choose the rates $\{R_t\}_{t=1}^{T}$
as well as the probe duration $T\probe$.
In doing so, we aim to maximize the sum data rate
$R\tot = \sum_{t=T\probe+1}^T R_t$ while satisfying the
expected error-probability constraint in \eqref{optimal_rate}.
Intuitively, we know that increasing $T\probe$ improves the SNR
estimate which, in turn, allows a higher data rate (since less
rate ``back-off'' is needed to satisfy the error constraint).
On the other hand, for a fixed block length $T$, the number of data
packets, $T-T\probe$, shrinks as $T\probe$ increases.
Therefore, the choice of $T\probe$ involves a tradeoff between these two
objectives.
In this section, we discuss the choice of $\{T\probe,R_1,\dots,R_T\}$ and
derive an upper bound on the sum rate $R\tot$ that leverages the
rate bounds from \secref{necessary_condition}
and the CRLB from \secref{CR-bound}.

In \secref{practical}, we recognized that the data-rate assignment problem
decouples in such a way that the optimal data rates
$\{R_t^*\}_{t=T_p+1}^T$ become independent of time $t$.
Thus, in the sequel, we focus on choosing a single data rate
$R\data$, whose optimal value will be denoted by $R\data^*$.
The system design problem then reduces to the following sum-rate
maximization:
\begin{align}
  \label{eq:sumrate_opt}
  R\tot^* \defn &
  \max_{T\probe\leqslant T,~
  (R_1,\dots,R_{T\probe},R\data)\in\mathcal{R}^{T\probe+1}} (T-T\probe) R\data
  	\\
  & \text{~~s.t.~~}
  \Ep{\varepsilon(\gamma,R\data)\ |\ \hat{\gamma}(\vec{I}_{T\probe})}
  \leqslant e^{-\alpha}. \nonumber
\end{align}
As argued in \secref{necessary_condition}, the optimal data rate
$R\data^*$ increases monotonically with the quality of the SNR estimate,
i.e., with the inverse of the estimator variance $1/\sigma_{N|\hat{\gamma}}^2$.
Thus, the optimal probe parameters $\{T\probe,R_1,\dots,R_{T\probe}\}$
are those that minimize $\sigma_{N|\hat{\gamma}}^2$.
From the CRLB in Theorem~\ref{th:CR-bound}, we know that
$\sigma_{N|\hat{\gamma}}^2 \geqslant
\underline{\sigma}_{N|\hat{\gamma}}^{2}(\gamma)$,
where
\begin{align}
\underline{\sigma}_{N|\hat{\gamma}}^{2}(\gamma)
&\defn
\min_{(R_1,\ldots ,R_{T\probe})\in \mathcal{R}^{T\probe}}
\left( \sum_{t=1}^{T\probe}
\frac{[\varepsilon'(\gamma,R_t)]^2}
{\varepsilon(\gamma,R_t)[1-\varepsilon(\gamma,R_t)]} \right)^{-1} \\
\label{eq:var_bnd1}
&= \left( \sum_{t=1}^{T\probe}
\max_{R_t \in \mathcal{R}} \frac{[\varepsilon'(\gamma,R_t)]^2}
{\varepsilon(\gamma,R_t)[1-\varepsilon(\gamma,R_t)]} \right)^{-1} .
\end{align}
Thus, if $\gamma$ was provided by a genie, and if the SNR estimator
was efficient (i.e., CRLB achieving), then \eqref{var_bnd1}
suggests to set the probe rate at
\begin{equation}
\label{eq:genieprobe}
R\probe\genie(\gamma)=\arg\max_{R_t\in\mathcal{R}}
\frac{[\varepsilon'(\gamma,R_t)]^2}
{\varepsilon(\gamma,R_t)[1-\varepsilon(\gamma,R_t)]},
\end{equation}
which is invariant to both time $t$ and probe duration $T_p$.
This yields
\begin{equation}
  \label{eq:var_bnd2}
  \underline{\sigma}_{N|\hat{\gamma}}^2(\gamma)
  = \frac{1}{T\probe}
    \frac{\varepsilon(\gamma,R\probe\genie(\gamma))
        [1-\varepsilon(\gamma,R\probe\genie(\gamma))]}
    {[\varepsilon'(\gamma,R\probe\genie(\gamma))]^2} .
\end{equation}
Using the genie-aided probe rate $R\probe\genie(\gamma)$, we can upper
bound the optimal sum rate \eqref{sumrate_opt} by
\begin{align}
  R\tot\genie 
  \defn&
  \max_{T\probe\leqslant T,~ R\data\in\mathcal{R}} (T-T\probe) R\data \\
  \label{eq:sumrate_bnd1}
  &\text{~~s.t.~~}
  \Ep{\varepsilon(\gamma,R\data)\ |\
  \hat{\gamma}(R\probe\genie(\gamma),T\probe)}
  \leqslant e^{-\alpha}, \nonumber
\end{align}
where we explicitly denote the dependence of the estimate $\hat{\gamma}$
on both $T\probe$ and $R\probe\genie(\gamma)$.

Next, recall that we established, in \secref{necessary_condition}, upper
bounds on the largest data rate that satisfies an
expected error constraint of the type in \eqref{sumrate_bnd1}.
In particular, \eqref{rate_bnd_QAM} gave an upper bound for
uncoded QAM signaling, and \eqref{rate_bnd_low_SNR} and
\eqref{rate_bnd_high_SNR} gave upper bounds for Gaussian signaling
in the low-SNR and high-SNR regimes, respectively.
These data-rate upper bounds,
$\bar{R}\data^*(\hat{\gamma},\sigma_{N|\hat{\gamma}}^2)$, can be applied
to \eqref{sumrate_bnd1} to bound the optimal sum rate as
$R\tot^* \leqslant \bar{R}\tot^*$, where
\begin{equation}
  \label{eq:sumrate_bnd2}
  \bar{R}\tot^* \defn
  \max_{T\probe\leqslant T}~ (T-T\probe)
    \bar{R}\data^*(\hat{\gamma},\sigma_{N|\hat{\gamma}}^2) ,
\end{equation}
and where $\hat{\gamma}$ and $\sigma_{N|\hat{\gamma}}^2$ are dependent
on both $T\probe$ and $R\probe(\gamma)$.
Since $\bar{R}\data^*(\hat{\gamma},\sigma_{N|\hat{\gamma}}^2)$ increases
monotonically in $1/\sigma_{N|\hat{\gamma}}^2$, we can upper bound
$\bar{R}\data^*$ using the lower bound on $\sigma_{N|\hat{\gamma}}^2$
established in \eqref{var_bnd2}.
This yields $\bar{R}\tot^*\leqslant R\tot^{\max}$ for
\begin{align}
  \label{eq:sumrate_max}
  R\tot^{\max}
  &\defn
  \max_{T\probe\leqslant T}~ (T-T\probe)
    \bar{R}\data^*(\hat{\gamma},
    \underline{\sigma}_{N|\hat{\gamma}}^2(\gamma)) .
\end{align}

\begin{figure*}[th!]
  \begin{center}
    \subfigure{
      \psfrag{Rt-n}[l][l][0.65]{$R_t^{\text{naive}}$}
      \psfrag{T=5}[l][l][0.65]{$\frac{1}{T}R\tot^{\max},~T=5$}
      \psfrag{T=50}[l][l][0.65]{$\frac{1}{T}R\tot^{\max},~T=50$}
      \psfrag{SNR (dB)}[][][0.8]{\sf $\hat{\gamma}$ (dB)}
      \label{fig:max_rate_QAM}
      \includegraphics[width=3.0in]{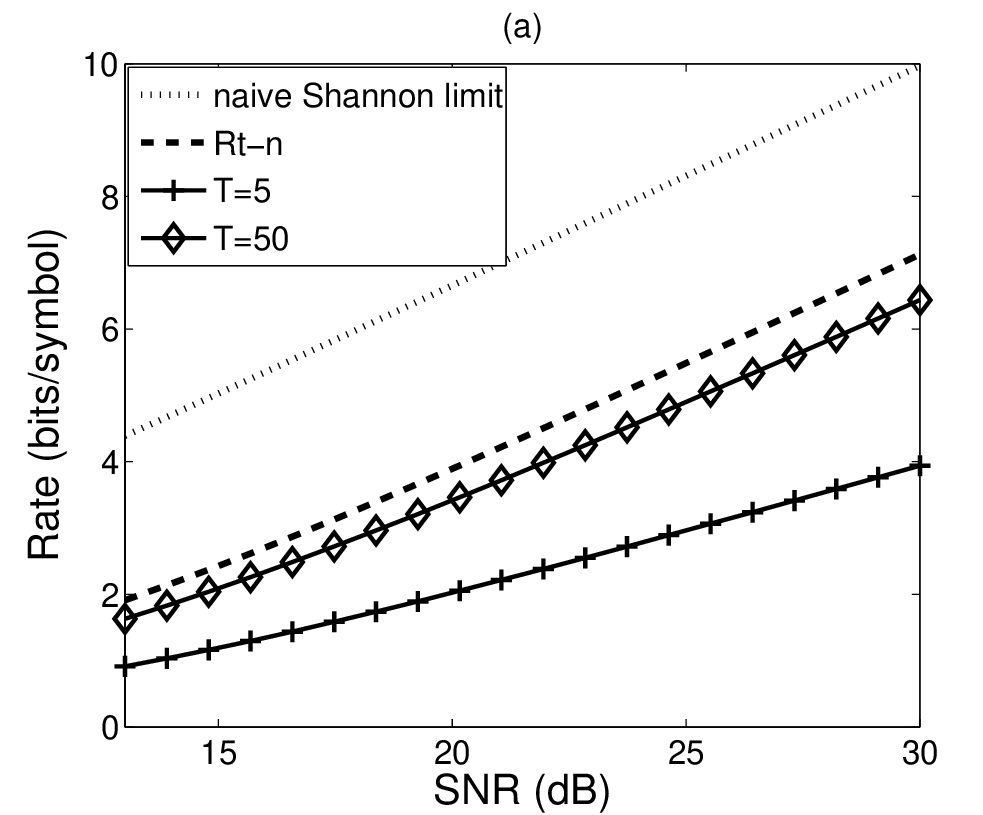}
    }
    \subfigure{
      \psfrag{Rt-n}[l][l][0.65]{$R_t^{\text{naive}}$}
      \psfrag{T=5}[l][l][0.65]{$\frac{1}{T}R\tot^{\max},~T=5$}
      \psfrag{T=50}[l][l][0.65]{$\frac{1}{T}R\tot^{\max},~T=50$}
      \psfrag{SNR (dB)}[][][0.8]{\sf $\hat{\gamma}$ (dB)}
      \label{fig:max_rate_Gaussian}
      \includegraphics[width=2.9in]{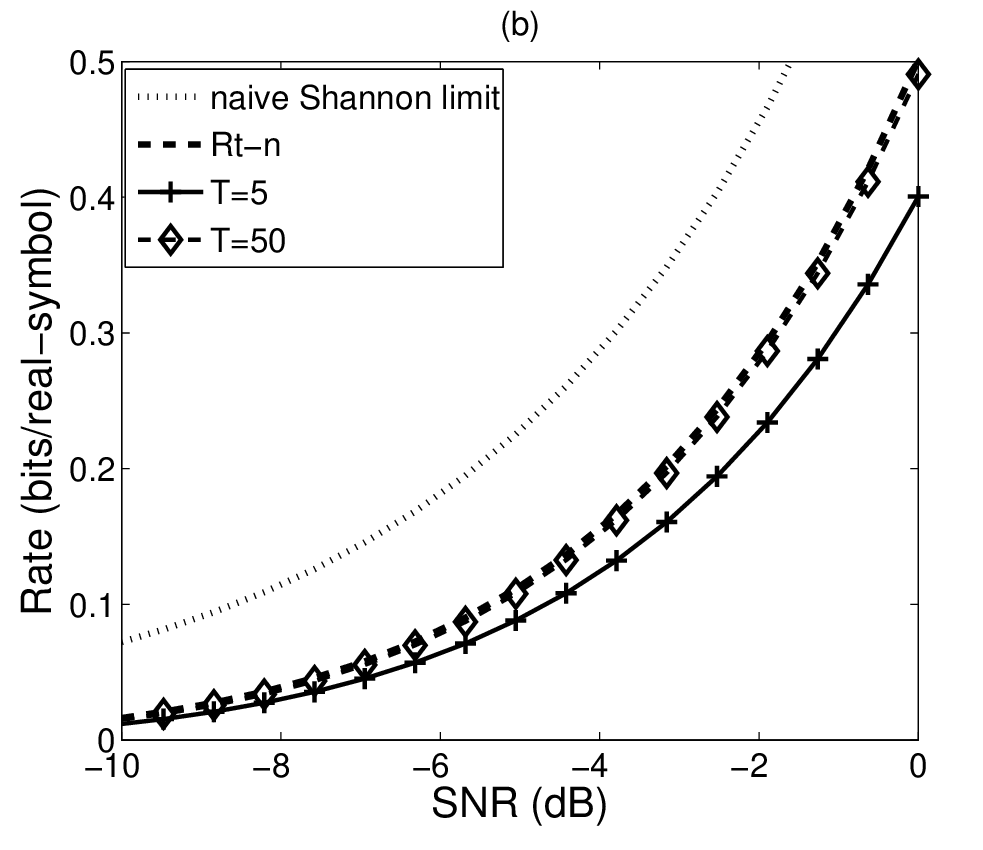}
    }
  \end{center}
  \caption{Normalized sum-rate bound
  $\frac{1}{T}R\tot^{\max}$ as a function of SNR $\hat{\gamma}$ for (a) QAM
  and (b) Gaussian signaling in the low-SNR regime.}
\end{figure*}

Figures~\ref{fig:max_rate_QAM} and \ref{fig:max_rate_Gaussian} plot
the normalized sum-rate bound $\frac{1}{T}R\tot^{\max}$
as a function of the estimated SNR $\hat{\gamma}$ for uncoded QAM and
Gaussian ensembles, respectively, at $T=5$ and $T=50$.
As before, we use target error rate $10^{-3}$ and packet size $n=500$.
For the genie-aided probe rate $R\probe\genie(\gamma)$ used to
calculate $\underline{\sigma}_{N|\hat{\gamma}}^2(\gamma)$, we assumed that
$\gamma \approx \hat{\gamma}$.
The figures also show $R\data^{\text{naive}}$ and the naive Shannon limit
$\frac{1}{2}\log_2(1+\hat{\gamma})$, for comparison.
Note that the difference between the naive rate $R\data^{\text{naive}}$
and the upper bound $\frac{1}{T}R\tot^{\max}$ increases significantly as
$T$ decreases.
This is due to the fact that, as $T$ decreases, it is too costly to allocate
a long probing interval, implying that the quality of SNR estimates decreases,
so that more rate back-off is required.
Note also that the difference between the naive rate and the upper bound
increases as the SNR increases.
This implies that the lack of perfect CSI becomes more
costly as the SNR increases.

%% file: estimator.tex
\section{An Asymptotically Optimal SNR Estimator}
\label{sec:estimator}


The quality of SNR estimates based on ACK/NAKs from a probe interval is
strongly dependent on both the probe rates $\{R_t\}_{t=1}^{T\probe}$
and the probe interval $T\probe$.
For the sum-rate upper bound derived in \secref{optimal_rate},
the probe rate $R\probe\genie(\gamma)$ in \eqref{genieprobe} was selected
in a genie-aided manner, assuming knowledge of the true SNR $\gamma$.
Clearly, $\gamma$ is not known in practice.

In this section, we develop a practical SNR estimator that, during the probing
interval $t\in\{1,\dots,T\probe\}$, \emph{recursively}
updates the probe rate $R_t$ and $\hat{\gamma}_t$ (i.e., the time-$t$
estimate\footnote{
  We emphasize that $\hat{\gamma}_t$ is the time-$t$ estimate of the
  time-invariant SNR $\gamma$, and should not be confused
  with the time-varying SNR $\gamma_t$ that was briefly used
  in \secref{model} before the time-invariance assumption was introduced.}
of $\gamma$) using the latest feedback pair $\{F_{t-1},R_{t-1}\}$.
We show that the probe rate adaptation is \emph{asymptotically optimal},
in that $R_t$ converges to $R\probe\genie(\gamma)$ for any initial probe
rate $R_1$.
Moreover, we show that our SNR estimator is {\em asymptotically efficient}
and {\em asymptotically normal}, i.e., that the corresponding estimation error
$N_t\defn\hat{\gamma}_t-\gamma$ converges to a zero-mean Gaussian random
variable whose variance is identical to the CRLB achieved with the
genie-aided probe rate $R\probe\genie(\gamma)$.
The normality of the error helps to justify the Gaussian approximation used
to derive the rate bounds (\ref{eq:condition_CR-bound_QAM}) and
(\ref{eq:condition_CR-bound_Gaussian}) for the uncoded QAM and Gaussian
cases, respectively.

\begin{quote}
\emph{The SNR Estimator:}
\begin{enumerate}
\item At time $t=1$, choose an arbitrary rate
$R_1\in \mathcal{R}$ and an arbitrary estimate $\hat{\gamma}_1$.
\item At each time $t=2,\ldots,T\probe$, update the estimate as
\begin{equation}
\label{eq:update_estimator}
\hat{\gamma}_t = \hat{\gamma}_{t-1} +
\frac{F_{t-1}-\varepsilon(\hat{\gamma}_{t-1},R_{t-1})}{(t-1)
\varepsilon'(\hat{\gamma}_{t-1},R_{t-1})} ,
\end{equation}
and choose the rate $R_t$ as:
\begin{equation}
\label{eq:rate_estimator}
R_t = \argmax_{R\in {\mathcal R}} \Phi(\hat{\gamma}_t,R) ,
\end{equation}
where $\Phi(\cdot,\cdot)$ is the Fisher information as defined in
(\ref{eq:fisher_info}).
\end{enumerate}
\end{quote}

We prove the following for our estimator.

\begin{theorem}
\label{theo:recursive_estimator}
For both uncoded QAM and Gaussian ensembles, as $T\probe \rightarrow \infty$,
\begin{equation}
\label{eq:estimator_convergence} \sqrt{T\probe} \left(
\hat{\gamma}_{T\probe}-\gamma \right) \stackrel{d}{\rightarrow}
N_{T_p} \sim {\mathcal N}\left(0,
\Phi^{-1}(\gamma,R\probe\genie(\gamma)) \right) .
\end{equation}
\end{theorem}

\begin{proof}
See Appendix~\ref{sec:appendix_2}.
\end{proof}

Theorem~\ref{theo:recursive_estimator} implies that our estimator
(\ref{eq:update_estimator}) is asymptotically efficient and consistent. Moreover, without any prior information on $\gamma$, rate allocation (\ref{eq:rate_estimator}) guarantees the performance achieved with the genie-aided probe rate $R\probe\genie(\gamma)$.
Next, we simulate the estimator.
Instead of the $t-1$ on the denominator, we use $(t-1)^{\beta}$ for various values of $\beta \in (0,1]$.

\begin{figure*}[th!]
  \begin{center}
    \subfigure[$R_t$ vs. $t$] {
      \label{fig:estimator_QAM_rate}
      \includegraphics[width=2in]{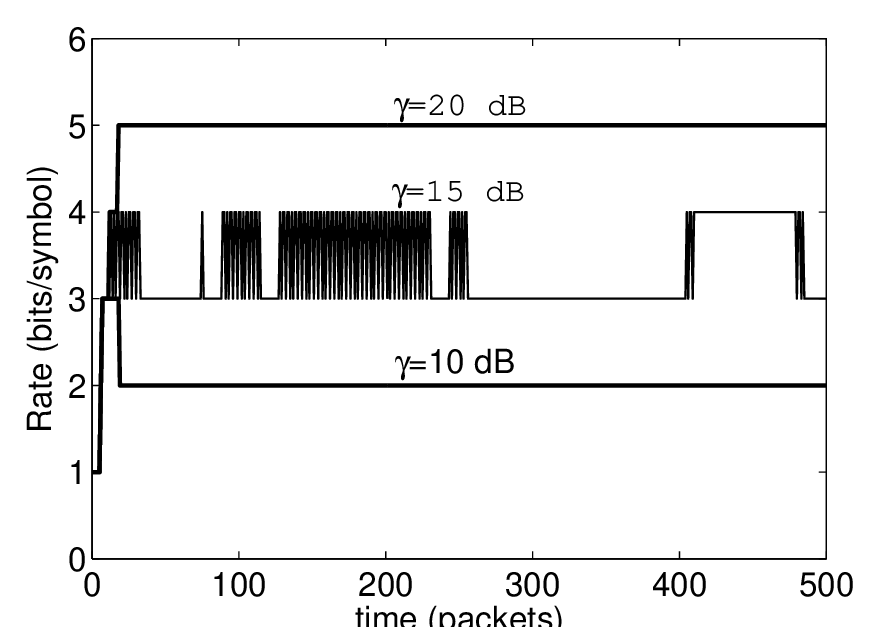}
    }
    \subfigure[$\hat{\gamma}_t$ vs. $t$ for $\beta=0.5$] {
      \label{fig:estimator_QAM_beta_0_5}
      \includegraphics[width=2in]{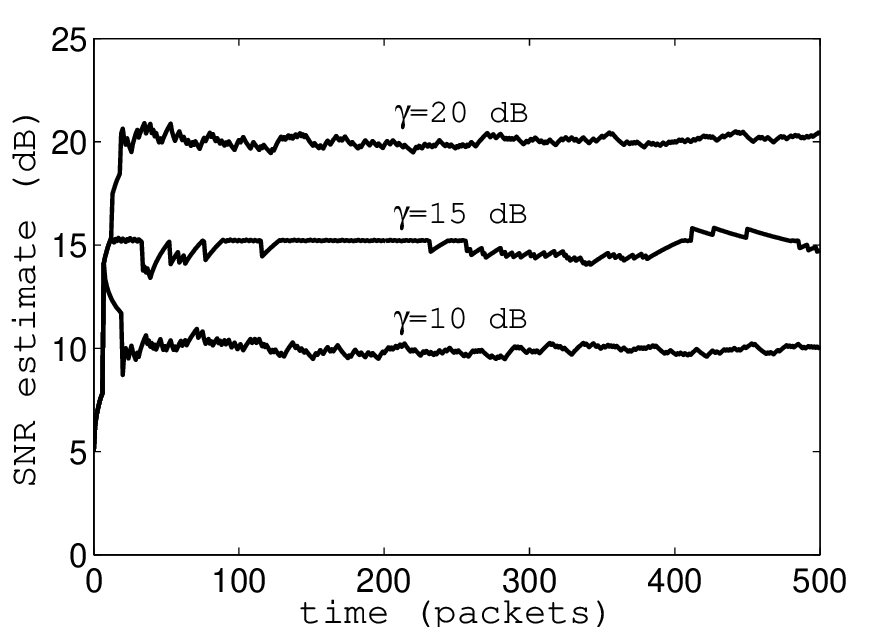}
    }
    \subfigure[$\hat{\gamma}_t$ vs. $t$ for $\beta=1$] {
      \label{fig:estimator_QAM_beta_1}
      \includegraphics[width=2in]{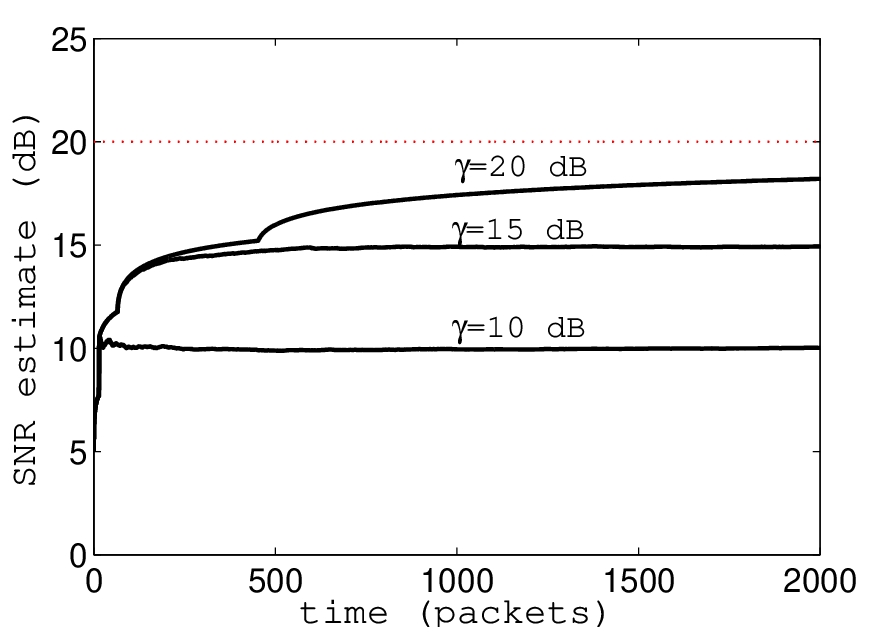}
    }
  \end{center}
  \caption{Example trajectories of the recursive SNR estimator when uncoded QAM is used.}
\label{fig:estimator_QAM}
\end{figure*}

In Fig.~\ref{fig:estimator_QAM}, a single realization of the estimator and the
corresponding assigned rate are illustrated for different values of $\gamma$,
over a block of $T\probe=500$ probe packets of size $n=500$ symbols.
The value of $\gamma$ and the asymptotic rate $R\probe\genie(\gamma)$
are also shown on the associated graphs.
The initial points for the estimator are
$\hat{\gamma}_1=3$~dB, $R_1=1$~bit/symbol, and the set of possible rates are
$\mathcal{R}=\{1,2,\ldots,10\}$ in bits/complex-symbol, i.e., the possible
constellation sizes are integer powers of $2$.
For $\beta=0.5$, one can observe that the optimal rate is reached with
approximately 20 probe packets for all values of SNR. Once that point is
reached, the estimation error variance decays fairly slowly due to the low
decay rate $\beta=0.5$. With a higher $\beta$, it takes longer to approach the
vicinity of $\gamma$, from the initial value $\hat{\gamma}_1$, but the
estimation error variance is lower once in steady state. This observation is
illustrated in Fig.~\ref{fig:estimator_QAM_beta_1}, where $\beta=1$ and the
probing block size is $T\probe=2000$ packets. In the realization corresponding to
$\gamma=20$ dB, the ``steady state'' is yet to be reached after 2000 packets.
On the other hand, the amplitude of the fluctuations around the final point
decay much faster, as one can observe in the realization corresponding to
$\gamma=10$~dB. Different choices for $\beta$ and the associated tradeoffs
involved in stochastic approximation algorithms are studied
in~\cite{Kushner:Book:97}.

\begin{figure*}[th!]
  \begin{center}
    \subfigure[$R_t$ vs. $t$] {
      \label{fig:estimator_Gaussian_rate}
      \includegraphics[width=2in]{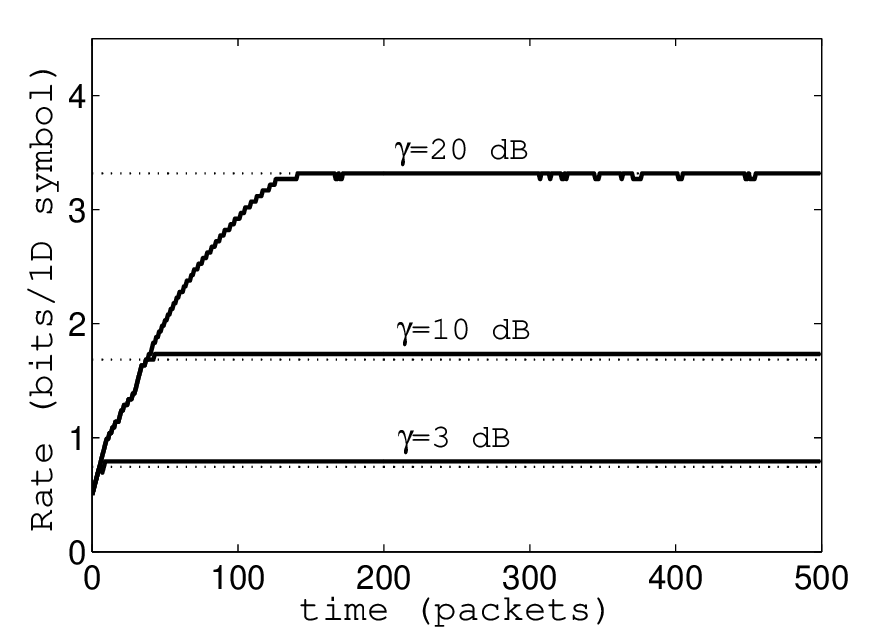}
    }
    \subfigure[$\hat{\gamma}_t$ vs. $t$ for $\beta=0.5$] {
      \label{fig:estimator_Gaussian_beta_0_5}
      \includegraphics[width=1.91in]{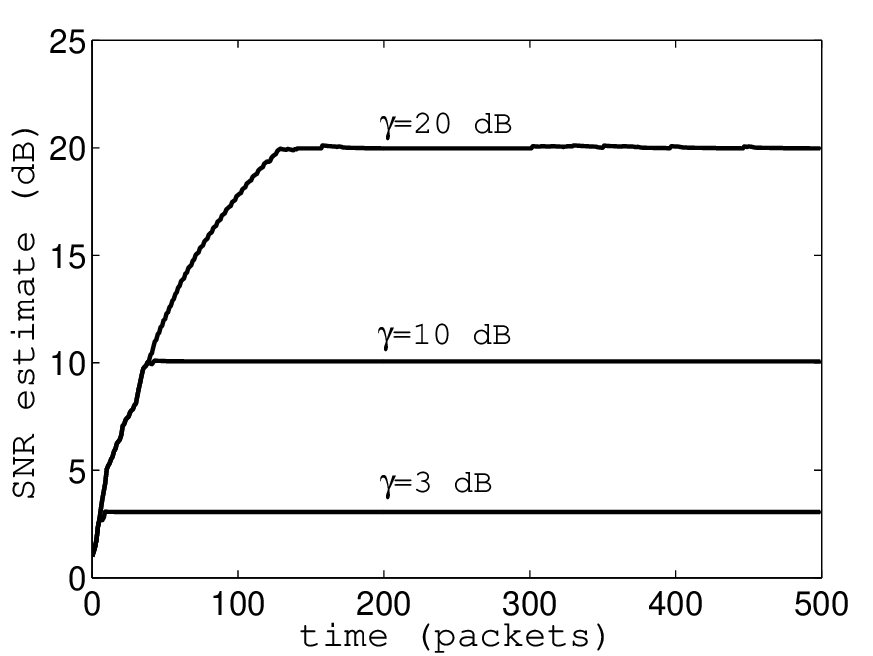}
    }
    \subfigure[$\hat{\gamma}_t$ vs. $t$ for $\beta=1$] {
      \label{fig:estimator_Gaussian_beta_1}
      \includegraphics[width=2in]{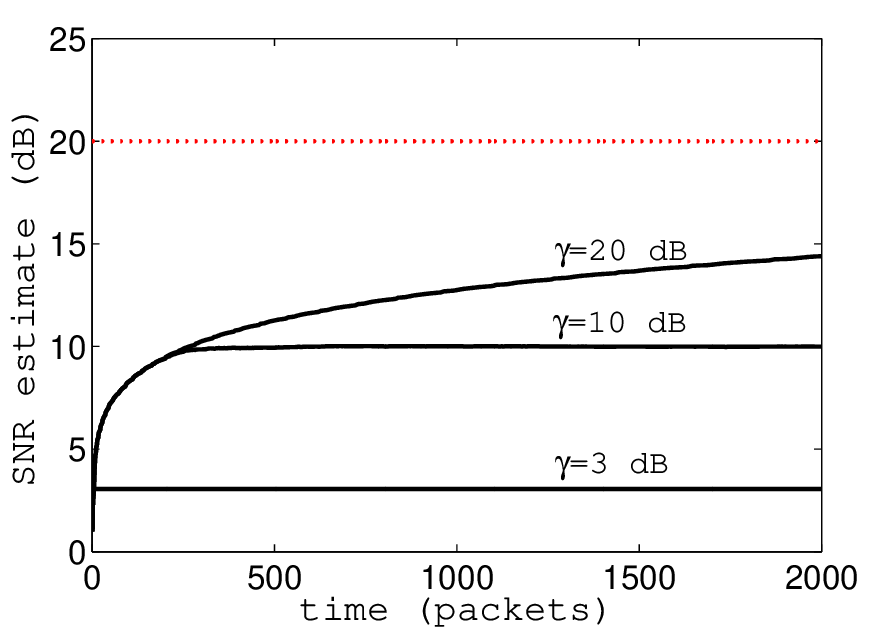}
    }
  \end{center}
  \caption{Example trajectories of the recursive SNR estimator when
    Gaussian signaling is used.}
\label{fig:estimator_Gaussian}
\end{figure*}

We illustrate our estimator response for Gaussian ensembles in
Fig.~\ref{fig:estimator_Gaussian}. As the set of rates ${\mathcal
R}$, we picked 100 points, equally spaced between 0 and 5
bits/real-symbol. The initial SNR estimate, $\hat{\gamma}_1=0$~dB,
was much smaller than the initial one in the QAM simulations, but
the initial rate, $R_1=0.5$~bits/complex-symbol, was identical to
the one in the QAM simulations. Here, we analyze SNR realizations
$\gamma=3,10$ and 20 dB.
With Gaussian ensembles, the convergence speed is slightly lower
than that with QAM. While the convergence is almost immediate for
$\gamma=3$ dB, it takes 30-40 packets for 10 dB and 130-140 packets
for $\gamma=20$ dB. This difference is mainly due to the difference
in the distances between the initial and final points. On the other
hand, due to the large size of the set of possible rates (unlike
QAM, where only a few discrete points are possible), there exists
some $R_t \in {\mathcal R}$ that is very close to the genie-aided
probe rate $R\probe\genie(\gamma)$. Consequently, the estimation
error variance decays much faster once $R_t$ comes near the vicinity
of $R\probe\genie(\gamma)$. We also illustrate the estimator with
$\beta=1$ in Fig.~\ref{fig:estimator_Gaussian_beta_1} and one can
notice the slow convergence, similar to the QAM simulations.

%% file: conclusion.tex
\section{Conclusion}	\label{sec:conc}

In this paper, we studied rate adaptation based on ACK/NAK feedback.  
In particular, we studied methods that maximize data rate subject to a constraint on expected packet-error probability, assuming that the transmitter has no knowledge of the SNR distribution.
Because optimal rate allocation was identified as a POMDP, which is impractical to implement, we focused on a suboptimal framework where a channel estimate is calculated based on previous feedback and a rate is chosen based on this channel estimate.
To aid the initial rate allocation, we allowed the use of $T\probe$ probe packets at the start of each data block.
First we considered a so-called ``naive'' rate allocator that maximizes rate subject to a constraint on \emph{instantaneous} packet-error probability, calculated from a given unbiased estimate $\hat{\gamma}$ of the true SNR $\gamma$.
Due to the inevitable error in SNR estimation, we argued that one must either back-off the naive rate, or correspondingly increase the SNR, to meet the stricter \emph{expected} error probability constraint. 
Based on a Gaussian approximation of the estimation error $N=\gamma-\hat{\gamma}$, we derived conditions on the ``effective estimator SNR'' $\hat{\gamma}^2/\sigma_{N|\hat{\gamma}}^2$ that are necessary for the existence of a feasible transmission rate, as well as an upper bound on the transmission rate when this necessary condition is satisfied. 
This latter analysis was carried out for both uncoded QAM signaling and random Gaussian signaling (the latter in both the low-SNR and high-SNR regimes). 
Next, we considered unbiased SNR estimation via ACK/NAK feedback. 
First, we lower bounded the error variance of those estimates (for general signaling schemes), and based on that bound, we lower bounded the necessary probing duration $T\probe$ and upper bounded the sum data rate (for both uncoded QAM signaling and random Gaussian signaling).
Finally, we proposed a practical unbiased ACK/NAK-based SNR estimator and showed that (as the probe duration increases) our estimator is asymptotically efficient and asymptotically normal.

%% file: training.tex
\section{Derivation of $T\probe^{\min}$ for uncoded QAM and Gaussian
Signaling}
\label{sec:appendix_1}

In this section, we derive \eqref{condition_CR-bound_QAM} and 
\eqref{condition_CR-bound_Gaussian}.
For brevity, we write
$\varepsilon \defn \varepsilon\probe(\gamma,R\probe)$ and
$\varepsilon' \defn \varepsilon'\probe(\gamma,R\probe)$.
Recall that, from (\ref{eq:condition_uncoded_QAM}) and 
(\ref{eq:CR-bound_fixed_rate}), we have for, uncoded QAM,
\begin{equation}
  T\probe^{\min}
  = \frac{\varepsilon(1-\varepsilon)}{(\varepsilon')^2}
  	\frac{2(\alpha+\ln 0.1 n)}{\hat{\gamma}^2}
\end{equation}
where, from \eqref{packet_err},
\begin{align}
  \varepsilon'
  &= \frac{\partial}{\partial \gamma} 
  	\bigg(1-\bigg[
	\underbrace{
	1- 0.2 \exp \left(-\frac{1.5 \gamma}{2^{R\probe}-1}\right)
	}_{(1-\varepsilon)^{1/n}}
  	\bigg]^n\bigg) \\
  &= n(1-\varepsilon)^{\frac{n-1}{n}} 
  	\underbrace{
	0.2 \exp \left(-\frac{1.5 \gamma}{2^{R\probe}-1}\right)
	}_{1-(1-\varepsilon)^{1/n}}
  	\left(-\frac{1.5}{2^{R\probe}-1}\right)  \\
  &= (1-\varepsilon)\big((1-\varepsilon)^{-1/n}-1\big) 
  	\underbrace{
  	\left(-\frac{1.5\gamma}{2^{R\probe}-1}\right) 
	}_{\ln\left(5\left(1-(1-\varepsilon)^{1/n}\right)\right)}
	\frac{n}{\gamma}.
\end{align}
Thus
\begin{align}
  T\probe^{\min}
  &= \frac{\varepsilon}{(1-\varepsilon)(\frac{\varepsilon'}{1-\varepsilon})^2}
  	\frac{2(\alpha+\ln 0.1 n)}{\hat{\gamma}^2} \\
  &= \frac{\varepsilon}{(1-\varepsilon)
  	\big((1-\varepsilon)^{-1/n}-1\big)^2
	\ln^2\left(5\left(1-(1-\varepsilon)^{1/n}\right)\right)}
	\nonumber\\&\quad\mbox{}\times
  	\frac{2(\alpha+\ln 0.1 n)}{(n\hat{\gamma}/\gamma)^2}.
\end{align}
From (\ref{eq:est_cond_low_SNR}) and (\ref{eq:CR-bound_fixed_rate}),
we have for, Gaussian signaling in the low-SNR regime,
\begin{equation}
  T\probe^{\min}
  = \frac{\varepsilon(1-\varepsilon)}{(\varepsilon')^2}
  	\frac{2\alpha}{\hat{\gamma}^2}
\end{equation}
where, from \eqref{error_bnd_random},
\begin{align}
  \varepsilon'
  &= \frac{\partial}{\partial \gamma} 
	\exp \left( n\rho^* \left[ R\probe\ln 2 -
	\frac{1}{2}\ln \left( 1+\frac{\gamma}{1+\rho^*} \right) \right] \right)\\
  &= \varepsilon \frac{-n\rho^*}{2}\frac{1}{(1+\frac{\gamma}{1+\rho^*})}
  	\frac{1}{1+\rho^*}
  = \varepsilon \frac{-n\rho^*}{2(1+\rho^*+\gamma)}.
\end{align}
Thus
\begin{align}
  T\probe^{\min}
  &= \frac{(1-\varepsilon)}{\varepsilon(\frac{\varepsilon'}{\varepsilon})^2}
  	\frac{2\alpha}{\hat{\gamma}^2} 
  = 8\alpha\frac{(1-\varepsilon)}{\varepsilon}
  	\frac{(1+\rho^*+\gamma)^2}{(n\rho^*\hat{\gamma})^2} .
\end{align}

%% file: estimator_convergence.tex
\section{Proof of Theorem~\ref{theo:recursive_estimator}}
\label{sec:appendix_2}

We will directly apply Theorem 2.1~\cite[p.~223]{Nevelson:Hasminskii:72}. The
necessary conditions for asymptotic normality and asymptotic efficiency to
hold in our system are:
\begin{enumerate}
\item The expectation, $\E{F_t}$, of observation $F_t$ must exist and must
be bounded:

$\E{F_t}=\varepsilon(\gamma,R_t)$ exists and is clearly bounded by 1 for all
$t$.

\item The partial derivative $\left| \frac{\partial \E{F_t}}{\partial \gamma} \right|$ must be
jointly continuous (in $\gamma$ and $R_t$) and bounded.

For both QAM~(\ref{eq:sym_error_uncoded}) and
Gaussian~(\ref{eq:error_bnd_random}) signals,
$|\frac{\partial\text{E}[F_t]}{\partial\gamma}|=
|\frac{\partial \varepsilon(\gamma,R_t)}{\partial\gamma}|$ is continuous and
bounded  for $\gamma \geqslant 0$ and $R_t \geqslant 0$.
\item The variance $\var{F_t}$ of observation $F_t$ must be continuous in
$\gamma$ and $R_t$.

For both QAM and Gaussian signaling,
$\var{F_t} = \varepsilon(\gamma,R_t) (1-\varepsilon(\gamma,R_t))$ is
continuous and bounded for  $\gamma \geqslant 0$ and $R_t \geqslant 0$.
\item Fisher information $\Phi(\gamma,R_t)$ must be continuous, positive and
for each $\gamma$, it must have a unique maximum in $R_t$.

For both QAM and Gaussian signaling, the Fisher information $\Phi(\gamma,R_t)$
as given in (\ref{eq:fisher_info}) is continuous and positive for
$\gamma \geqslant 0$ and $R_t \geqslant 0$. Moreover, it has a
unique maximum $R_t=R\probe\genie(\gamma)$ for each $\gamma \geqslant 0$, since
$\Phi(\gamma,R_t)$ is a strictly concave and continuous function of $R_t$.

\item For some $b > 2$, $\E{|F_t|^b}$ must be bounded for all possible values
of $\gamma$ and associated rate $R\probe\genie(\gamma)$.

Since $F_t\in\{0,1\}$, we know that $\E{|F_t|^b}$ is bounded for all $b > 2$
and for all values of $(\gamma,R\probe\genie(\gamma))$.

Furthermore, the asymptotic
efficiency~\cite[p.~186,224]{Nevelson:Hasminskii:72} of the estimator is
\begin{align*}
\lefteqn{ \left. \Phi(\gamma,R_t)\cdot \frac{\left( \frac{\partial}{\partial \gamma}
\E{F_t}\right)^2}{\var{F_t}} \right|_{R_t=R\probe\genie(\gamma)} }\nonumber\\
&=
\Phi(\gamma,R\probe\genie(\gamma))\cdot
\frac{\left(\varepsilon'(\gamma,R\probe\genie(\gamma))\right)^2}{\varepsilon
(\gamma,R\probe\genie(\gamma))(1-\varepsilon(\gamma,R\probe\genie(\gamma)))} \\
&=1.
\end{align*}
The asymptotic optimality, i.e., $T\probe
\sigma_{N_{T\probe}|\hat{\gamma}_{T\probe}}^2 \rightarrow \left[
\Phi(\gamma,R\probe\genie(\gamma))\right]^{-1}$ as $T\probe
\rightarrow \infty$ follows as a consequence of Theorem
2.1~\cite[p.~223]{Nevelson:Hasminskii:72}.
\end{enumerate}